\documentclass[11pt]{article}
\def\spacingset#1{\def\baselinestretch{#1}\small\normalsize}
\usepackage{subfigure}
\usepackage{graphics}
\usepackage{graphicx}
\usepackage{amsbsy}
\usepackage{amssymb}
\usepackage{amscd}
\usepackage{amsthm}
\usepackage{amsmath}
\usepackage{subfigure}
\usepackage{enumerate}
\usepackage{xcolor}
\usepackage{color}
\usepackage{amsmath,epsfig,fancyhdr, amsthm, amssymb, enumerate}
\usepackage[ruled]{algorithm2e}

\newcommand{\n}{{\cal N} \, }
\newcommand{\ba}{\begin{array}}
\newcommand{\ea}{\end{array}}
\newcommand{\be}{\begin{displaymath}}
\newcommand{\ee}{\end{displaymath}}
\newcommand{\ben}{\begin{equation}}
\newcommand{\een}{\end{equation}}
\newcommand{\bena}{\begin{eqnarray}}
\newcommand{\eena}{\end{eqnarray}}
\newcommand{\beqa}{\begin{eqnarray*}}
\newcommand{\enqa}{\end{eqnarray*}}
\newcommand{\f}{\frac}
\newcommand{\bc}{\begin{center}}
\newcommand{\ec}{\end{center}}
\newcommand{\bi}{\begin{itemize}}
\newcommand{\ei}{\end{itemize}}
\newcommand{\benu}{\begin{enumerate}}
\newcommand{\eenu}{\end{enumerate}}
\newcommand{\bdes}{\begin{description}}
\newcommand{\edes}{\end{description}}
\newcommand{\bt}{\begin{tabular}}
\newcommand{\et}{\end{tabular}}

\newcommand{\Phibf}{\mbox{${\bf \Phi}$}}

\newcommand \nubf{\mbox{\boldmath$\nu$\unboldmath}}

\newcommand \bbf{{\bf b}}
\newcommand \cbf{{\bf c}}
\newcommand \dbf{{\bf d}}

\newcommand \fbf{{\bf f}}
\newcommand \gbf{{\bf g}}
\newcommand \hbf{{\bf h}}

\newcommand \pbf{{\bf p}}

\newcommand \ubf{{\bf u}}
\newcommand \vbf{{\bf v}}
\newcommand \wbf{{\bf w}}
\newcommand \xbf{{\bf x}}
\newcommand \ybf{{\bf y}}
\newcommand \zbf{{\bf z}}

\newcommand \Dbf{{\bf D}}

\newcommand \Fbf{{\bf F}}

\newcommand \Ibf{{\bf I}}

\newcommand \Pbf{{\bf P}}
\newcommand \Qbf{{\bf Q}}

\newcommand \Wbf{{\bf W}}

\newcommand \Cset{{\cal C}}

\newcommand{\Rnum}{{\mathbb R}}

\newcommand{\Nnum}{{\mathbb N}}

\newcommand{\Ncal}{{\cal N}}

\newcommand{\zebf}{{\bf 0}}

\newcommand{\circlambda}{\mbox{$\Lambda$
             \kern-.85em\raise1.5ex
             \hbox{$\scriptstyle{\circ}$}}\,}

\def\Pr{\mathop{\rm Pr}}

\newtheorem{Theorem}{Theorem}

\newtheorem{Lemma}{Lemma}
\newtheorem{Corollary}[Theorem]{Corollary}
\begin{document}
\title{\vspace{-0.5cm}
Compressed Sensing with General Frames via Optimal-dual-based
$\ell_1$-analysis }
\author{Yulong Liu, Tiebin Mi and Shidong Li}
\author{Yulong Liu\thanks{Yulong Liu is with
the Institute of Electronics, Chinese Academy of Sciences, Beijing,
$100190$, China. Email: \{yulong3.liu@gmail.com\}.}, Tiebin
Mi\thanks{Tiebin Mi is with the School of Information Sciences,
Renmin University of China, $100872$, China. Email:
\{mitiebin@gmail.com\}.},  and Shidong Li\thanks{Shidong Li is with
the Renmin University of China; and the Department of Mathematics,
San Francisco State University, San Francisco, CA 94132, USA. Email:
\{shidong@sfsu.edu.\}}}

\date{}

\maketitle
 \pssilent \setcounter{page}{1} \thispagestyle{empty}
\begin{abstract}
\spacingset{1.8} Compressed sensing with sparse frame
representations is seen to have much greater range of practical
applications than that with orthonormal bases. In such settings, one
approach to recover the signal is known as $\ell_1$-analysis. We
expand in this article the performance analysis of this approach by
providing a weaker recovery condition than existing results in the
literature. Our analysis is also broadly based on general frames and
alternative dual frames (as analysis operators). As one application
to such a general-dual-based approach and performance analysis, an
optimal-dual-based technique is proposed to demonstrate the
effectiveness of using alternative dual frames as analysis
operators. An iterative algorithm is outlined for solving the
optimal-dual-based $\ell_1$-analysis problem.  The effectiveness of
the proposed method and algorithm is demonstrated through several
experiments.

{\bf Keywords:} Compressed sensing, $\ell_1$-synthesis,
$\ell_1$-analysis, optimal-dual-based $\ell_1$-analysis,  frames,
dual frames, Bregman iteration, split Bregman iteration.

\end{abstract}
\spacingset{1.9}
\section{Introduction}
Compressed sensing concerns the problem of recovering a
high-dimensional sparse signal from a small number of linear
measurements
\begin{equation}
  \label{eq1} \ybf = \Phibf \fbf + \zbf,
\end{equation}
where $\Phibf$ is an $m \times n$ sensing matrix with $m \ll n$ and
$\zbf \in \Rnum^{m}$ is a noise term modeling measurement error. The
goal is to reconstruct the unknown signal $\fbf \in \Rnum^{n}$ based
on available measurements $\ybf \in \Rnum^{m}$. References on
compressed sensing have a long list, including, e.g.,
\cite{Candes2006b, Candes2006c, Candes2006, Donoho2006,
Donoho2006a}.

In standard compressed sensing scenarios, it is usually assumed that
$\fbf$ has a sparse (or nearly sparse) representation in an
orthonormal basis. However, a growing number of applications in
signal processing point to problems where $\fbf$ is sparse with
respect to an overcomplete dictionary or a frame rather than an
orthonormal basis, see, e.g., \cite{Mallat1993}, \cite{Chen2001}, \cite{Bruckstein2009}, and references therein. Examples
include, e.g., signal modeling in array signal processing
(oversampled array steering matrix), reflected radar and sonar
signals (Gabor frames), and images with curves (curvelets), etc. The
flexibility of frames is the key characteristic that empowers frames
to become a natural and concise signal representation tool.
Compressed sensing, with works including, e.g., \cite{Rauhut2008}, \cite{Candes2010}, that deals with sparse representations with
respect to frames becomes therefore particularly important.  In this
setting the signal $\fbf$ is expressed as $\fbf=\Dbf\xbf$ where
$\Dbf\in \Rnum^{n\times d}$ ($n < d$) is a matrix of frame vectors
(as columns) that are often rather coherent in applications, and
$\xbf\in \Rnum^{d}$ is a sparse coefficient vector.  The linear
measurements of $\fbf$ then become
\begin{equation} \label{eqn_fisADx}
\ybf = \Phibf \Dbf\xbf + \zbf.
\end{equation}

Since $\xbf$ is assumed sparse, a straightforward way of recovering
$\fbf$ from (\ref{eqn_fisADx}) is known as $\ell_{1}$-synthesis (or
synthesis-based method) \cite{Chen2001},
\cite{Elad2007}, \cite{Candes2010}. One first finds the sparsest possible coefficient
$\xbf$ by solving an $\ell_{1}$ minimization problem
\begin{equation}
\label{eq2} \hat{\xbf}=\underset {\tilde{\xbf}\in\Rnum^{d}}
{\textrm{argmin}} \Vert\tilde{\xbf}\Vert_{1} \ \ \ \  s.t. \ \
\Vert\ybf-\Phibf\Dbf\tilde{\xbf}\Vert_{2} \leq \epsilon,
\end{equation}
where $\Vert\xbf\Vert_{p} \ (p=1,2)$ denotes the standard
$\ell_{p}$-norm of the vector $\xbf$ and $\epsilon^2$ is a likely
upper bound on the noise power $\|\zbf\|_2^2$. Then the solution to
$\fbf$ is derived via a synthesis operation, i.e.,
$\hat{\fbf}=\Dbf\hat{\xbf}$.

Although empirical studies show that $\ell_{1}$-synthesis often
achieves good recovery results, little is known about the theoretical performance of this method.
The analytical results in \cite{Rauhut2008} essentially require that
the frame $\Dbf$ has columns that are extremely uncorrelated such
that $\Phibf\Dbf$ satisfies the requirements imposed by the
traditional compressed sensing assumptions. However, these
requirements are often infeasible when $\Dbf$ is highly coherent.
For example, consider a simple case in which $\Phibf\in
\Rnum^{m\times n}$ is a Gaussian matrix with i.i.d. entries, then
$\Phibf\sim\n(\zebf, \Ibf_{n}\otimes \Ibf_{m})$, where $\otimes$
denotes the Kronecker product and $\Ibf_{m}$ is an identity matrix
of the size $m$.  It is now well known that with very high
probability $\Phibf$ has small $s$-restricted isometry constant when
$m$ is on the order of $s\log(n/s)$ \cite{Candes2006b}, \cite{Baraniuk2008}.
Let us now examine $\Phibf\Dbf$.   It is not
hard to show that $\Phibf\Dbf\sim\n(\zebf,
\Dbf^{*}\Dbf\otimes\Ibf_{m})$, where $(\cdot)^{*}$ denotes the
transpose operation. Consequently, if $\Dbf$ is a coherent frame, $\Phibf\Dbf$ does not
generally satisfy the common restricted isometry property (RIP)
\cite{Rauhut2008}. Meantime, the mutual incoherence property (MIP)
\cite{Donoho2006a} may not apply either, as it is very hard for
$\Phibf\Dbf$ to satisfy the MIP as well when $\Dbf$ is highly
correlated.
\par
The analysis-based method, $\ell_{1}$-analysis, is an alternative to
$\ell_{1}$-synthesis, e.g., \cite{Elad2005},
\cite{Elad2007}, \cite{Candes2010}, which finds the estimate $\hat{\fbf}$ directly by
solving the problem
\begin{equation}\label{eq3}
  \hat{\fbf}=\underset {\tilde{\fbf}\in\Rnum^{n}} {\textrm{argmin}} \Vert\Dbf^{*}\tilde{\fbf}\Vert_{1} \ \ \ \  s.t. \ \
  \Vert\ybf-\Phibf\tilde{\fbf}\Vert_{2} \leq \epsilon.
 \end{equation}

When $\Dbf$ is a basis, the $\ell_{1}$-analysis and the
$\ell_{1}$-synthesis approaches are equivalent. However, when $\Dbf$
is an overcomplete frame, it was observed that there is a recovery
performance gap between them \cite{Chen2001}, \cite{Elad2007}. No
clear conclusion has been reached as to which approach is better
without specifying applications and associated data sets.
\par
A performance study of the $\ell_{1}$-analysis approach is just
recently given in \cite{Candes2010}. It was shown that \eqref{eq3}
recovers a signal $\hat{\fbf}$ with an error bound
\begin{equation}
  \label{eq4} \|\hat{\fbf}-\fbf\|_{2} \leq C_{0}\cdot \epsilon +
  C_{1}\cdot\f{\|\Dbf^{*}\fbf-(\Dbf^{*}\fbf)_{s}\|_{1}}{\sqrt{s}},
\end{equation}
provided that $\Phibf$ obeys a $\Dbf$-RIP condition (see
\eqref{eq8}) with $\delta_{2s}<0.08$, where the columns of $\Dbf$
form a Parseval frame and $(\Dbf^{*}\fbf)_{s}$ is a vector
consisting of the largest $s$ entries of $\Dbf^{*}\fbf$ in
magnitude. It follows from \eqref{eq4} that if $\Dbf^{*}\fbf$ has
rapidly decreasing coefficients, then the solution to \eqref{eq3} is
very accurate. In particular, if the measurements of $\fbf$ are
noiseless and $\Dbf^{*}\fbf$ is exactly $s$-sparse, then $\fbf$ is
recovered exactly.

Indeed, $\ell_{1}$-analysis shows a promising performance in
applications where both the columns of the Gram matrix
$\Dbf^{*}\Dbf$ and the coefficient vector $\xbf$ are reasonably
sparse,  see e.g., \cite{Elad2007}, \cite{CaiJ2009}, \cite{Candes2010}.
In other words, as long as the frame coefficient
vector $\Dbf^{*}\fbf$ is sensibly sparse, $\ell_{1}$-analysis can be
the right method to use.
\par
However, the $\ell_1$-analysis approach of (\ref{eq3}) is certainly
not flawless.  That $\fbf$ is sparse in terms of $\Dbf$ does not
imply $\Dbf^{*}\fbf$ is necessarily sparse. In fact, as the
canonical dual frame expansion in the case of Parseval frames,
$\Dbf^{*}\fbf=\Dbf^{*}\Dbf\xbf$ has the minimum $\ell_{2}$ norm by
the frame property, see, e.g., \cite{Christensen2003} and is usually
fully populated which is also pointed out in \cite{Rauhut2008}.

For a given signal $\fbf$, there are infinitely many ways to
represent $\fbf$ by the columns of $\Dbf$. By the spirit of frame
expansions, all coefficients of a frame expansion of $\fbf$ in
$\Dbf$ should correspond to some dual frame of $\Dbf$. It is not
hard to imagine that there should be some dual frame of $\Dbf$,
denoted by $\widetilde{\Dbf}$, such that $\widetilde{\Dbf}^*\fbf$ is
sparser than $\Dbf^{*}\fbf$. Furthermore, if a similar error bound
(just like \eqref{eq4}) holds for arbitrary dual frame analysis
operators, then one may expect a better recovery performance by
taking some ``proper'' dual frame of $\Dbf$ as the analysis
operator. Motivated by this observation, we consider a
\textit{general-dual-based $\ell_1$-analysis} as follows:
\begin{equation}
 \label{generall1analysis}
 \hat{\fbf}=\underset {\ \tilde{\fbf}\in\Rnum^{n}} {\textrm{argmin}}
 \Vert\widetilde{\Dbf}^{*}\tilde{\fbf}\Vert_{1} \ \ \ \  s.t. \ \
  \Vert\ybf-\Phibf\tilde{\fbf}\Vert_{2} \leq \epsilon,
\end{equation}
where columns of the analysis operator $\widetilde{\Dbf}$ form a
general (and any) dual frame of $\Dbf$.

In this article, we first present a performance analysis for the
general-dual-based $\ell_1$-analysis approach
\eqref{generall1analysis}. It turns out that a recovery error bound
exists entirely similar to that of \eqref{eq4}. More precisely,
under suitable conditions, \eqref{generall1analysis} recovers a
signal $\hat{\fbf}$ with an error bound
\begin{equation}
  \label{errorboundforgenernalL1analysis} \|\hat{\fbf}-\fbf\|_{2} \leq C_{0}\cdot \epsilon +
  C_{1}\cdot\f{\|\widetilde{\Dbf}^{*}\fbf-(\widetilde{\Dbf}^{*}\fbf)_{s}\|_{1}}{\sqrt{s}}.
\end{equation}
We show that sufficient conditions which ratify a recovery
performance estimation \eqref{errorboundforgenernalL1analysis}
depend not only on the $\Dbf$-RIP of $\Phibf$, but also on the ratio
of frame bounds. By utilizing the Shifting Inequality
\cite{Cai2010a}, the recovery condition on the sensing matrix is
improved from $\delta_{2s}<0.08$ \cite{Candes2010} to
$\delta_{2s}<0.2$ under the same assumptions that columns of $\Dbf$
form a Parseval frame and $\widetilde{\Dbf}=\Dbf$.

The important question then is how to choose some appropriate dual
frame $\widetilde{\Dbf}$ such that $\widetilde{\Dbf}^*\fbf$ is as
sparse as possible.  One approach as we propose here is by the
method of \textit{optimal-dual-based $\ell_1$-analysis}:
\begin{equation}
 \label{eq5}
 \hat{\fbf}=\underset {\Dbf\widetilde{\Dbf}^* = \Ibf, \ \tilde{\fbf}\in\Rnum^{n}} {\textrm{argmin}}
 \Vert\widetilde{\Dbf}^{*}\tilde{\fbf}\Vert_{1} \ \ \ \  s.t. \ \
  \Vert\ybf-\Phibf\tilde{\fbf}\Vert_{2} \leq \epsilon,
\end{equation}
where the optimization is not only over the signal space but also
over all dual frames of $\Dbf$. Note that the class of all dual
frames for $\Dbf$ is given by \cite{Li1995} (see
\eqref{generaldualII})
\begin{equation} \label{generaldual}
 \widetilde{\Dbf}  = (\Dbf\Dbf^*)^{-1}\Dbf +
 \Wbf^*(\Ibf_d-\Dbf^*(\Dbf\Dbf^*)^{-1}\Dbf)= \bar{\Dbf} + \Wbf^*\Pbf,
\end{equation}
where $\bar{\Dbf}\equiv(\Dbf\Dbf^*)^{-1}\Dbf$ denotes the canonical
dual frame of $\Dbf$, $\Pbf\equiv \Ibf_d
-\Dbf^*(\Dbf\Dbf^*)^{-1}\Dbf$ is the orthogonal projection onto the
null space of $\Dbf$, and $\Wbf \in \Rnum^{d \times n}$ is an
arbitrary matrix. Plug \eqref{generaldual} into \eqref{eq5}, we
obtain
\begin{equation}
 \label{generaldualbasedL1analysis}
 \hat{\fbf}=\underset {\tilde{\fbf}\in\Rnum^{n},\  \gbf\in\Rnum^{d}} {\textrm{argmin}}
 \Vert\bar{\Dbf}^{*}\tilde{\fbf} + \Pbf \gbf\Vert_{1} \ \ \ \  s.t. \ \
  \Vert\ybf-\Phibf\tilde{\fbf}\Vert_{2} \leq \epsilon,
\end{equation}
where we have used the fact that when $\tilde{\fbf}\neq \zebf$,
$\gbf\equiv\Wbf\tilde{\fbf}$ can be any vector in $\Rnum^{d}$ due to
the fact that $\Wbf$ is free.

Clearly, the solution to \eqref{generaldualbasedL1analysis}
definitely corresponds to that of \eqref{generall1analysis} with
some optimal dual frame, say $\widetilde{\Dbf}_{\text{o}}$ as the
analysis operator. The optimality here is in the sense that
$\|\widetilde{\Dbf}_{\text{o}}^*\hat{\fbf}\|_1$ achieves the
smallest $\|\widetilde{\Dbf}^*\tilde{\fbf}\|_1$ in value among all
dual frames $\widetilde{\Dbf}$ of $\Dbf$ and feasible signals
$\tilde{\fbf}$ satisfied the constraint in
\eqref{generaldualbasedL1analysis}. When $\fbf$ is sparse with
respect to $\Dbf$, it is highly desirable that the corresponding
optimal dual frame should be effective in sparsfying the true signal
$\fbf$. It then follows from \eqref{errorboundforgenernalL1analysis}
that an accurate recovery of $\fbf$ may be achieved by the solution
of \eqref{generaldualbasedL1analysis}. Indeed, we have seen that the
signal recovery via \eqref{generaldualbasedL1analysis} is much more
effective than that of the $\ell_1$-analysis approach \eqref{eq3}
which uses the canonical dual frame as the analysis operator.

Finally, we also develop an iterative algorithm for solving the
optimal-dual-based $\ell_1$-analysis problem. The proposed algorithm
is based on the split Bregman iteration introduced in
\cite{Goldstein2009}. Our numerical results show that the proposed
algorithm is very fast when properly chosen parameter values are
used.

This paper is organized as follows.  Section \ref{section2} contains preliminary discussions about compressed sensing with general frames.  Performance studies for the general-dual-based
$\ell_1$-analysis approach are presented in section \ref{section3}. In section
\ref{section4}, an optimal-dual-based $\ell_1$-analysis approach and a corresponding iterative algorithm are discussed. In section \ref{section5}, results of numerical experiments are presented to illustrate
the effectiveness of signal recovery via the optimal-dual-based
$\ell_1$-analysis approach.  Conclusion remarks are given in
section \ref{section6}.  Included in the appendix is on the basics of the Bregman iteration which is beneficial to the discussion of the algorithm presented in section \ref{section4}.

\section{Preliminaries} \label{section2}
\subsection{Preliminaries for Compressed Sensing}
Let $\xbf \in \Rnum^{d}$ be a column vector. The \textit{support} of
$\xbf$ is defined as $\textrm{supp}(\xbf) = \{i:\xbf_{i}\neq 0,
i=1,\ldots, d\}$. For $s \in \Nnum$, a vector $\xbf$ is said to be
$s$-\textit{sparse} if $\vert \textrm{supp}(\xbf)\vert\leq s$. For
$T \subseteq \{1, \ldots, d\}$, $\xbf_{T}$ stands for a $|T|$-long
vector taking entries from $\xbf$ indexed by $T$. Similarly,
$\Dbf_{T}$ is the submatrix of $\Dbf$ restricted to the columns
indexed by $T$. We shall write $\Dbf_{T}^{*} \equiv (\Dbf_{T})^{*}$,
and use the standard notation $\Vert\xbf\Vert_{q}$ to denote the
$\ell_{q}$-norm of $\xbf$
$$\|\xbf\|_{q}=\left\{ \begin{array}{cc}
\left(\sum_{i=1}^{n}|\xbf_{i}|^{q}\right)^{1/q} & 1\leq q < \infty,
\\ \underset {1\leq i \leq n} {\textrm{max}} |\xbf_{i}|  & q = \infty.
\end{array}            \right.$$

For an $m\times n$ measurement matrix $\Phibf$, we say that $\Phibf$
obeys the restricted isometry property \cite{Candes2005} with
constant $\gamma_{s}\in (0,1)$ if
\begin{equation}\label{eq7}
(1-\gamma_{s})\Vert \xbf \Vert_{2}^{2}\leq \Vert \Phibf\xbf
\Vert_{2}^{2} \leq(1+\gamma_{s})\Vert \xbf \Vert_{2}^{2}
\end{equation}
holds for all $s$-sparse signals $\xbf$. We say that $\Phibf$
satisfies the restricted isometry property adapted to $\Dbf$
(abbreviated $\Dbf$-RIP) \cite{Candes2010} with constant
$\delta_{s}\in (0,1)$ if
\begin{equation}\label{eq8}
(1-\delta_{s})\Vert \vbf \Vert_{2}^{2}\leq \Vert \Phibf\vbf
\Vert_{2}^{2} \leq(1+\delta_{s})\Vert \vbf \Vert_{2}^{2}
\end{equation}
holds for all $\vbf\in \Sigma_{s}$, where $\Sigma_{s}$ is the union
of all subspaces spanned by all subsets of $s$ columns of $\Dbf$.
Obviously, $\Sigma_{s}$ is the image under $\Dbf$ for all $s$-sparse
vectors.  Similar to $\gamma_{s}$, it is easy to see $\delta_{s}$ is
monotone, i.e., $\delta_{s}\leq \delta_{s_{1}}$, if $s\leq s_{1}\leq
d$.

The $\Dbf$-RIP condition is also validated in a number of
discussions. For instance, it was shown in \cite{Candes2010}
  that suppose an $m \times n$ matrix $\Phibf$ obeys a concentration
  inequality of the type
\begin{equation}
  \label{ConcentrationInequality}
  \Pr\left( \left|\|\Phibf\nubf\|_2^2 - \|\nubf\|_2^2 \right| \geq
  \delta \|\nubf\|_2^2 \right) \leq
  c e^{-\gamma \delta^2 m}, \ \ \ \ \delta \in (0, 1)
\end{equation}
for any fixed $\nubf \in \Rnum^n$, where $\gamma$, $c$ are some
positive constants, then $\Phibf$ will satisfy the $\Dbf$-RIP
(associated with some $\Dbf$-RIP constant) with overwhelming
probability provided that $m$ is on the order of $s\log(d/s)$. Many
types of random matrices satisfy \eqref{ConcentrationInequality},
some examples include matrices with Gaussian, subgaussian, or
Bernoulli entries. Very recently, it has also been shown in
\cite{Krahmer2011} that randomizing the column signs of any matrix
that satisfies the standard RIP results in a matrix which satisfies
the Johnson-Lindenstrauss lemma \cite{Johnson1984}. Such a matrix
would then satisfy the $\Dbf$-RIP via
\eqref{ConcentrationInequality}. Consequently, partial Fourier
matrix (or partial circulant matrix) with randomized column signs
will satisfy the $\Dbf$-RIP since these matrices are known to
satisfy the RIP.

\subsection{Preliminaries for Frame Theory}
A set of vectors $\{\dbf_{k}\}_{k\in I}$ in
$\Rnum^{n}$ is a \textit{frame} of $\Rnum^{n}$ if there exist
constants $0<A\leq B<\infty$ such that
\begin{equation}
  \label{eq9}
  \forall \  \fbf \in \Rnum^{n}, \ \ \ \ A\|\fbf\|_{2}^{2} \leq \sum_{k\in I}|\langle\fbf, \dbf_{k}
  \rangle|^{2} \leq B\|\fbf\|_{2}^{2},
\end{equation}
where numbers $A$ and $B$ are called \textit{frame bounds}. A frame
that is \textit{not} a basis is said to be \textit{overcomplete} or
\textit{redundant}. More details about frames can be found in e.g.,
\cite{Christensen2003}, \cite{Han2007}, \cite{Heil1989}. In the
matrix form, \eqref{eq9} can be reformulated as
\begin{equation}
  \label{eq10}
   \forall \  \fbf \in \Rnum^{n}, \ \ \ \ A\|\fbf\|_{2}^{2} \leq \fbf^{*}(\Dbf\Dbf^{*})\fbf \leq B\|\fbf\|_{2}^{2},
\end{equation}
where $\{\dbf_{k}\}_{k\in I}$ are the columns of $\Dbf$. When
$A=B=1$, the columns of $\Dbf$ form a Parseval frame and
$\Dbf\Dbf^{*}=\Ibf$. A frame $\{\widetilde{\dbf}_{k}\}_{k\in I}$ is
an \textit{alternative dual frame} of $\{\dbf_{k}\}_{k\in I}$ if
\begin{equation}\label{eq11}
  \forall \  \fbf \in \Rnum^{n}, \ \ \ \ \fbf = \sum_{k\in I}\langle\fbf, \widetilde{\dbf}_{k}
  \rangle\  \dbf_{k} = \sum_{k\in I}\langle\fbf, {\dbf}_{k}
  \rangle\  \widetilde{\dbf}_{k}.
\end{equation}
For every given overcomplete frame $\{\dbf_{k}\}_{k\in I}$, there
are infinite many dual frames $\{\widetilde{\dbf}_{k}\}_{k\in I}$
such that \eqref{eq11} holds \cite{Li1995}. More precisely, the
class of all dual frames for $\Dbf$ is given by the columns of
$\widetilde{\Dbf}$
\begin{equation} \label{generaldualII}
 \widetilde{\Dbf}  = (\Dbf\Dbf^*)^{-1}\Dbf +
 \Wbf^*(\Ibf_d-\Dbf^*(\Dbf\Dbf^*)^{-1}\Dbf)
 = (\Dbf\Dbf^*)^{-1}\Dbf + \Wbf^*\Pbf.
\end{equation}
Note that $\Dbf\widetilde{\Dbf}^* = \Ibf$. When $\Wbf = \zebf$,
$\widetilde{\Dbf}$ reduces to the canonical dual frame
$\bar{\Dbf}=(\Dbf\Dbf^{*})^{-1}\Dbf$. The lower and upper frame
bound of $\bar{\Dbf}$ is given by $1/B$ and $1/A$, respectively. For
$\fbf \in \Rnum^{n}$, the canonical coefficients
$\bar{\Dbf}^{*}\fbf$ have the minimum $\ell_{2}$ norm, i.e.,
$\|\bar{\Dbf}^*\fbf\|_2 = \underset {\tilde{\xbf}:
\Dbf\tilde{\xbf}=\fbf} {\textrm{min}} \|\tilde{\xbf}\|_2$.

\subsection{The Shifting Inequality}

We now briefly discuss the Shifting Inequality \cite{Cai2010a},
which is a very useful tool performing finer estimation of
quantities involving $\ell_{1}$ and $\ell_{2}$ norms. A different
proof of this inequality is also given in \cite{Foucart2010}.
\begin{Lemma} \label{lem1}
(Shifting Inequality {\rm \cite{Cai2010a}}) Let $q$, $r$ be positive
integers satisfying $q\leq 3r$. Then any nonincreasing sequence of
real numbers $a_{1}\geq\cdots\geq a_{r}\geq b_{1}\geq\cdots\geq
b_{q}\geq c_{1}\geq\cdots\geq c_{r}\geq 0$ satisfies
\begin{equation}
\label{eq13}\sqrt{\sum_{i=1}^{q}b_{i}^{2}+\sum_{i=1}^{r}c_{i}^{2}}\leq
\f{\sum_{i=1}^{r}a_{i}+\sum_{i=1}^{q}b_{i}}{\sqrt{q+r}}.
\end{equation}
\end{Lemma}

\section{Sufficient Conditions for General-dual-based
$\ell_1$-analysis} \label{section3}
In this section, we establish theoretical results for the
general-dual-based $\ell_1$-analysis approach
\eqref{generall1analysis} in which the analysis operator can be any
dual frame of $\Dbf$. Our main result is that, under suitable
conditions, the solution to \eqref{generall1analysis} is very
accurate provided that $\widetilde{\Dbf}^*\fbf$ has rapidly
decreasing coefficients. We present two results with slightly
different emphasises. They are, respectively, when the analysis
operator is an alternative dual frame and when the analysis operator
is the canonical dual frame.

\subsection{The Case of Alternative Dual Frames}
\begin{Theorem}\label{thm1}
  Let $\Dbf$ be a general frame of $\Rnum^{n}$ with frame bounds $0<A\leq B<\infty$.  Let $\widetilde{\Dbf}$ be an alternative
  dual frame of $\Dbf$ with frame bounds $0<\widetilde{A}\leq \widetilde{B}<\infty$, and let $\rho=s/b$. Suppose
  \begin{equation}
    \label{eq14} \left(1-\sqrt{\rho B \widetilde{B}}\right)^2 \cdot \delta_{s+a} +
\rho B \widetilde{B}\cdot\delta_{b} < 1 - 2\sqrt{\rho B
\widetilde{B}}
  \end{equation}
holds for some positive integers $a$ and $b$ satisfying $0< b-a\leq
3a$. Then the solution $\hat{\fbf}$ to \eqref{generall1analysis}
satisfies
  \begin{equation}
    \label{eq15} \Vert \hat{\fbf}-\fbf \Vert_{2} \leq C_{0}\cdot\epsilon +
    C_{1}\cdot\f{\Vert\widetilde{\Dbf}^{*}\fbf-(\widetilde{\Dbf}^{*}\fbf)_{s}\Vert_{1}}{\sqrt{s}},
  \end{equation}
 where $C_{0}$ and $C_{1}$ are some constants and $(\widetilde{\Dbf}^{*}\fbf)_{s}$
 denotes the vector consisting the largest $s$ entries of
 $\widetilde{\Dbf}^{*}\fbf$ in magnitude.
\end{Theorem}

\begin{proof} The proof is inspired by that of \cite{Candes2006a}.
Let $\fbf$ and $\hat{\fbf}$ be as in the theorem. Set $\hbf =
\fbf-\hat{\fbf}$. Our goal is to bound the norm of $\hbf$. Without
loss of generality, we assume that the first $s$ entries of
$\widetilde{\Dbf}^{*}\fbf$ are the largest in magnitude. Making
rearrangement if necessary, we may also assume that
$$\vert(\widetilde{\Dbf}^{*}\hbf)(s+1)\vert\geq
\vert(\widetilde{\Dbf}^{*}\hbf)(s+2)\vert\geq\cdots,$$ where
$(\widetilde{\Dbf}^{*}\hbf)(k)$ denotes the $k$th component of
$\widetilde{\Dbf}^{*}\hbf$. Let $T_{0}=\{1, 2, \ldots, s\}$. In
order to apply the Shifting Inequality, we partition $T_{0}^{c}$
(complement set of $T_{0}$) into the following sets: $T_{1}=\{s+1,
s+2, \ldots, s+a\}$ and $T_{i}=\{s+a+(i-2)b+1, \ldots,
s+a+(i-1)b\}$, $i=2, 3, \ldots$, with the last subset of size less
than or equal to $b$, where $a$ and $b$ are positive integers
satisfying $0<b-a\leq 3a$. Further divide each $T_{i}, i\geq 2$ into
two pieces. Set
$$T_{i1}=\{s+a+(i-2)b+1, \cdots, s+(i-1)b\},$$
and
$$T_{i2}=T_{i}\backslash T_{i1} =\{s+(i-1)b+1, \cdots,
s+(i-1)b+a\}.$$ Note that $|T_{i1}|=b-a$ and $|T_{i2}|=a$ for all
$i\geq 2$. For simplicity, we denote $T_{01}=T_{0}\cup T_{1}$. Note
first that
\begin{eqnarray}
\nonumber \Vert\hbf\Vert_{2} & = &
\Vert\Dbf\widetilde{\Dbf}^{*}\hbf\Vert_{2}=
\Vert\Dbf_{T_{01}}\widetilde{\Dbf}_{T_{01}}^{*}\hbf +
\Dbf_{T_{01}^{c}}\widetilde{\Dbf}_{T_{01}^{c}}^{*}\hbf\Vert_{2}\\\nonumber
& \leq & \| \Dbf_{T_{01}}\widetilde{\Dbf}_{T_{01}}^{*}\hbf\|_2
+ \| \Dbf_{T_{01}^{c}}\widetilde{\Dbf}_{T_{01}^{c}}^{*}\hbf \|_2 \\
\label{eq16} & \stackrel{\eqref{eq10}}{\leq} & \|
\Dbf_{T_{01}}\widetilde{\Dbf}_{T_{01}}^{*}\hbf\|_2 + \sqrt{B}\|
\widetilde{\Dbf}_{T_{01}^{c}}^{*}\hbf \|_2,
\end{eqnarray}
where $\widetilde{\Dbf}_{T}^{*}\equiv(\widetilde{\Dbf}_{T})^{*}$. To
bound the norm of $\hbf$, it is required to bound
$\Vert\widetilde{\Dbf}_{T_{01}^{c}}^{*}\hbf\Vert_{2}$ and
$\Vert\Dbf_{T_{01}}\widetilde{\Dbf}_{T_{01}}^{*}\hbf\Vert_{2}$. Then
the proof proceeds in following three steps:\\
\textbf{Step 1: Bound the tail
$\Vert\widetilde{\Dbf}_{T_{01}^{c}}^{*}\hbf\Vert_{2}$}. \
Since $\fbf$ and $\hat{\fbf}$ are feasible and $\hat{\fbf}$ is the
minimizer, we have
\begin{eqnarray}
  \nonumber \Vert\widetilde{\Dbf}_{T_{0}}^{*}\fbf \Vert_{1}+ \Vert\widetilde{\Dbf}_{T_{0}^{c}}^{*}\fbf \Vert_{1}=\Vert\widetilde{\Dbf}^{*}\fbf \Vert_{1} & \geq & \Vert\widetilde{\Dbf}^{*}\hat{\fbf}
  \Vert_{1} = \Vert\widetilde{\Dbf}^{*}\fbf - \widetilde{\Dbf}^{*}\hbf \Vert_{1}\\\nonumber
  & = & \Vert\widetilde{\Dbf}_{T_{0}}^{*}\fbf-\widetilde{\Dbf}_{T_{0}}^{*}\hbf \Vert_{1}+\Vert\widetilde{\Dbf}_{T_{0}^{c}}^{*}\fbf -
  \widetilde{\Dbf}_{T_{0}^{c}}^{*}\hbf\Vert_{1}\\ \nonumber & \geq &
  \Vert\widetilde{\Dbf}_{T_{0}}^{*}\fbf\Vert_{1} -
  \Vert\widetilde{\Dbf}_{T_{0}}^{*}\hbf\Vert_{1} +
  \Vert\widetilde{\Dbf}_{T_{0}^{c}}^{*}\hbf\Vert_{1} -
  \Vert\widetilde{\Dbf}_{T_{0}^{c}}^{*}\fbf\Vert_{1}.
\end{eqnarray}
This implies
\begin{equation}
  \label{eq17} \Vert\widetilde{\Dbf}_{T_{0}^{c}}^{*}\hbf\Vert_{1} \leq
  \Vert\widetilde{\Dbf}_{T_{0}}^{*}\hbf\Vert_{1} +
  2\Vert\widetilde{\Dbf}_{T_{0}^{c}}^{*}\fbf\Vert_{1}.
\end{equation}
If $0< b-a \leq 3a$, then applying the Shifting Inequality
\eqref{eq13} to the vectors
$\left[(\widetilde{\Dbf}_{T_{1}}^{*}\hbf)^{*},
(\widetilde{\Dbf}_{T_{21}}^{*}\hbf)^{*},
(\widetilde{\Dbf}_{T_{22}}^{*}\hbf)^{*} \right]^{*}$ and
$\left[(\widetilde{\Dbf}_{T_{(i-1)2}}^{*}\hbf)^{*},
(\widetilde{\Dbf}_{T_{i1}}^{*}\hbf)^{*},
(\widetilde{\Dbf}_{T_{i2}}^{*}\hbf)^{*} \right]^{*}$ for $i = 3, 4,
\ldots$, we have
\begin{eqnarray*}
\Vert\widetilde{\Dbf}_{T_{2}}^{*}\hbf\Vert_{2} &\leq&
\f{\Vert\widetilde{\Dbf}_{T_{1}}^{*}\hbf\Vert_{1}+\Vert\widetilde{\Dbf}_{T_{21}}^{*}\hbf\Vert_{1}}{\sqrt{b}},
 \cdots,
\\ \Vert\widetilde{\Dbf}_{T_{i}}^{*}\hbf\Vert_{2}  &\leq&
\f{\Vert\widetilde{\Dbf}_{T_{(i-1)2}}^{*}\hbf\Vert_{1}+\Vert\widetilde{\Dbf}_{T_{i1}}^{*}\hbf\Vert_{1}}{\sqrt{b}},
 \cdots.
\end{eqnarray*}
It then follows that
\begin{eqnarray*}
\sum_{i\geq2}\Vert\widetilde{\Dbf}_{T_{i}}^{*}\hbf\Vert_{2}& \leq &
\f{\Vert\widetilde{\Dbf}_{T_{0}^{c}}^{*}\hbf \Vert_{1}}{\sqrt{b}}\\
& \stackrel{\eqref{eq17}}{\leq} &
\f{\Vert\widetilde{\Dbf}_{T_{0}}^{*}\hbf \Vert_{1}}{\sqrt{b}} +
\f{2\Vert\widetilde{\Dbf}_{T_{0}^{c}}^{*}\fbf \Vert_{1}}{\sqrt{b}}\\
& \stackrel{C.S.}{\leq} &
\sqrt{\f{s}{b}}\Vert\widetilde{\Dbf}_{T_{0}}^{*}\hbf
\Vert_{2}+\f{2\Vert\widetilde{\Dbf}_{T_{0}^{c}}^{*}\fbf \Vert_{1}}{\sqrt{b}} \\
& = & \sqrt{\rho}(\Vert\widetilde{\Dbf}_{T_{0}}^{*}\hbf
\Vert_{2}+\eta)\\
& \stackrel{\eqref{eq10}}{\leq} &
\sqrt{\rho}\left({\sqrt{\widetilde{B}}}\|\hbf\|_{2}+\eta\right),
\end{eqnarray*}
where $\rho = {s}/{b}$, $\eta =
{2\Vert\widetilde{\Dbf}_{T_{0}^{c}}^{*}\fbf \Vert_{1}}/{\sqrt{s}}$,
and $C.S.$ stands for the Cauchy-Schwarz inequality. Hence,
$\Vert\widetilde{\Dbf}_{T_{01}^{c}}^{*}\hbf\Vert_{2}$ is bounded by
\begin{equation}
\label{eq18}\Vert\widetilde{\Dbf}_{T_{01}^{c}}^{*}\hbf\Vert_{2} \leq
\sum_{i\geq2}\Vert\widetilde{\Dbf}_{T_{i}}^{*}\hbf\Vert_{2} \leq
\sqrt{\rho}\left({\sqrt{\widetilde{B}}}\|\hbf\|_{2}+\eta\right).
\end{equation}\\
\textbf{Step 2: Show
$\Vert\Dbf_{T_{01}}\widetilde{\Dbf}_{T_{01}}^{*}\hbf\Vert_{2}$ is
appropriately small}. \ On the one hand,
\begin{equation}\label{eq19}\Vert\Phibf\hbf\Vert_{2}=\Vert\Phibf\fbf-\ybf-(\Phibf\hat{\fbf}-\ybf)\Vert_{2}
\leq
\Vert\Phibf\fbf-\ybf\Vert_{2}+\Vert\Phibf\hat{\fbf}-\ybf\Vert_{2}\leq
2\epsilon.
\end{equation}
On the other hand,
\begin{eqnarray} \nonumber
  \Vert\Phibf\hbf\Vert_{2} = \Vert\Phibf\Dbf\widetilde{\Dbf}^{*}\hbf\Vert_{2} & =
  &\|\Phibf\Dbf_{T_{01}}\widetilde{\Dbf}_{T_{01}}^{*}\hbf + \Phibf\Dbf_{T_{01}^{c}}\widetilde{\Dbf}_{T_{01}^{c}}^{*}\hbf
  \|_{2}\\ \nonumber & \geq &
  \Vert\Phibf\Dbf_{T_{01}}\widetilde{\Dbf}_{T_{01}}^{*}\hbf
  \Vert_{2} - \sum_{i\geq 2}\Vert\Phibf\Dbf_{T_{i}}\widetilde{\Dbf}_{T_{i}}^{*}\hbf
  \Vert_{2} \\ \nonumber& \stackrel{\eqref{eq8}}{\geq} & \sqrt{1-\delta_{s+a}}\Vert\Dbf_{T_{01}}\widetilde{\Dbf}_{T_{01}}^{*}\hbf
  \Vert_{2} - \sqrt{1+\delta_{b}}\sum_{i\geq 2}\Vert\Dbf_{T_{i}}\widetilde{\Dbf}_{T_{i}}^{*}\hbf
  \Vert_{2} \\ \nonumber& \geq & \sqrt{1-\delta_{s+a}}\Vert\Dbf_{T_{01}}\widetilde{\Dbf}_{T_{01}}^{*}\hbf
  \Vert_{2} - \sqrt{1+\delta_{b}}\sum_{i\geq 2}\Vert\Dbf_{T_{i}}\|_{2}\|\widetilde{\Dbf}_{T_{i}}^{*}\hbf
  \Vert_{2}\\ \nonumber & \stackrel{\eqref{eq10}}{\geq} & \sqrt{1-\delta_{s+a}}\Vert\Dbf_{T_{01}}\widetilde{\Dbf}_{T_{01}}^{*}\hbf
  \Vert_{2} - \sqrt{(1+\delta_{b})B}\sum_{i\geq 2}\|\widetilde{\Dbf}_{T_{i}}^{*}\hbf
  \Vert_{2}\\ \label{eq20}
  & \stackrel{\eqref{eq18}}{\geq} & \sqrt{1-\delta_{s+a}}\Vert\Dbf_{T_{01}}\widetilde{\Dbf}_{T_{01}}^{*}\hbf
  \Vert_{2} - \sqrt{\rho(1+\delta_{b})B}\left({\sqrt{\widetilde{B}}}\|\hbf\|_{2}+\eta\right).
\end{eqnarray}
Combining \eqref{eq19} and \eqref{eq20} yields
\begin{equation}
  \label{eq21} \sqrt{1-\delta_{s+a}}\Vert\Dbf_{T_{01}}\widetilde{\Dbf}_{T_{01}}^{*}\hbf
  \Vert_{2} \leq 2\epsilon + \sqrt{\rho(1+\delta_{b})B}\left({\sqrt{\widetilde{B}}}\|\hbf\|_{2}+\eta\right).
\end{equation}\\
\textbf{Step 3: Bound the error of $\hbf$}. \  It
follows from \eqref{eq16} and \eqref{eq18},
\begin{eqnarray}
\nonumber\|\hbf\|_{2} & \leq & \|
\Dbf_{T_{01}}\widetilde{\Dbf}_{T_{01}}^{*}\hbf\|_2 + \sqrt{B}\|
\widetilde{\Dbf}_{T_{01}^{c}}^{*}\hbf \|_2 \\ \label{eq22} &
\stackrel{\eqref{eq18}}{\leq} & \|
\Dbf_{T_{01}}\widetilde{\Dbf}_{T_{01}}^{*}\hbf\|_2 + \sqrt{\rho B
\widetilde{B}}\|\hbf\|_2+\sqrt{\rho B}\cdot\eta.
\end{eqnarray}
Combining \eqref{eq21} with \eqref{eq22} yields
\begin{equation}
  \label{eq24} K_{1}\Vert\hbf\Vert_{2}\leq 2\epsilon +
  K_{2}\eta,
\end{equation}
where \begin{eqnarray*} K_{1}& = & \sqrt{1-\delta_{s+a}}-\sqrt{\rho
 B \widetilde{B}(1-\delta_{s+a})}
-\sqrt{\rho B
\widetilde{B}(1+\delta_{b})},\\
K_{2} & = & \sqrt{\rho B(1-\delta_{s+a})}+\sqrt{\rho
B(1+\delta_{b})}.
\end{eqnarray*}
If $K_{1}$ is positive, then we have
\begin{equation}
  \label{eq25} \Vert\hbf\Vert_{2}\leq \f{2}{K_{1}}\cdot\epsilon +
  \f{K_{2}}{K_{1}}\cdot\eta = C_{0}\cdot\epsilon +
    C_{1}\cdot\f{\Vert\widetilde{\Dbf}\fbf-(\widetilde{\Dbf}\fbf)_{s}\Vert_{1}}{\sqrt{s}},
\end{equation}
where $C_{0}={2}/{K_{1}}$ and $C_{1}={2K_{2}}/{K_{1}}$. At last,
note that if
\begin{equation}
    \label{eq27} \left(1-\sqrt{\rho B
\widetilde{B}}\right)^2 \cdot \delta_{s+a} + \rho B
\widetilde{B}\cdot\delta_{b} < 1 - 2\sqrt{\rho B \widetilde{B}},
  \end{equation}
then $K_1>0$. This completes the proof.
\end{proof}

\noindent{\bf Remark 1:}\ The $\Dbf$-RIP condition can now be
$\delta_{2s} < 0.1398$ in the case of Parseval frames.   Suppose
$\Dbf$ is a Parseval frame and the analysis operator
$\widetilde{\Dbf}$ is its canonical dual frame, i.e.,
$\widetilde{\Dbf}=\Dbf$ as seen in \cite{Candes2010}. Then
\eqref{eq14} becomes, since $B \widetilde{B}=1$,
\begin{equation}\label{eq28}
(1-\sqrt{\rho})^{2}\cdot\delta_{s+a} + \rho\cdot\delta_{b} <
1-2\sqrt{\rho}.
\end{equation}
Note that different choices of $a$ and $b$ may lead to different
conditions. For example, let $a=3s, b=12s,$ and $\rho = s/b =1/12$.
Then \eqref{eq28} becomes
\begin{equation}\label{eq29}
 (13-4\sqrt{3})\cdot\delta_{4s}+\delta_{12s}<12-4\sqrt{3}.
\end{equation}
By the fact that $\delta_{ks}\leq k\cdot\delta_{2s}$ for positive
integers $k$ and $s$ (Corollary 3.4 of \cite{Needell2009}),
\eqref{eq29} is satisfied whenever
$\delta_{2s}<(3-\sqrt{3})/(16-4\sqrt{3})\approx 0.1398$. This
condition is weaker than the condition $\delta_{2s}<0.08$ obtained
in \cite{Candes2010}.

\noindent {\bf Remark 2:}\ When $\Dbf$ is a general frame and the
analysis operator $\widetilde{\Dbf}$ is its canonical dual frame,
i.e., $\widetilde{\Dbf}=(\Dbf\Dbf^{*})^{-1}\Dbf$, then \eqref{eq14}
may be expressed as
\begin{equation}\label{eq30}
(1-\sqrt{\rho\kappa})^{2}\cdot\delta_{s+a} + \rho\kappa\cdot\delta_{b} <
1-2\sqrt{\rho\kappa},
\end{equation}
where $\kappa=B \widetilde{B}=B/A$ is the ratio of the frame bounds.
We see that this sufficient condition not only depends on the
$\Dbf$-RIP constants of $\Phibf$, but also on the ratio of frame
bounds $\kappa=B/A$. Furthermore, as $\kappa$ increases, it will lead to
a stronger condition on $\Phibf$. For instance, let $a=7s, b=8s$,
and $\rho=1/8$, for different $\kappa$'s, e.g., $\kappa=1$ and
$\kappa=\sqrt{2}$, \eqref{eq30} becomes $\delta_{8s}<0.5395$ and
$\delta_{8s}<0.3104$, respectively. The former is obviously much
weaker than the latter. Hence, from this point of view, whenever a
Parseval frame is allowed in specific applications, it makes sense
to use the Parseval frame $(\kappa=1)$.

\noindent {\bf Remark 3:}\ In general, when $\Dbf$ is a general
frame and $\widetilde{\Dbf}$ is an alternative dual frame of $\Dbf$,
we see that the product of the upper frame bounds $B\widetilde{B}$
(of $\Dbf$ and $\widetilde{\Dbf}$) is a factor in the sufficient
condition. Evidently, $B\widetilde{B}$ is similar to $\kappa$ in the
case of the canonical dual. A larger $B\widetilde{B}$ will lead to a
stronger condition on $\Phibf$.

\noindent {\bf Remark 4:}\ The results obtained in Theorem
\ref{thm1} for bounded noise can be applied directly to Gaussian
noise, i.e., $\zbf\sim \n(\zebf, \sigma^2\Ibf_m)$, because in this
case $\zbf$ belongs to a bounded set with large probability, as the
following lemma asserted.
\begin{Lemma}{\rm \cite{Cai2009}}\label{lem2}
The Gaussian error $\zbf\sim \Ncal(\zebf, \sigma^{2}\Ibf_{m})$
satisfies
\begin{equation}
 \label{eq47} \Pr\left(\Vert\zbf\Vert_{2}\leq \sigma\sqrt{m+2\sqrt{m\log
  m}} \right) \geq  1-\f{1}{m}.
\end{equation}
\end{Lemma}
A combination of Theorem \ref{thm1} and Lemma \ref{lem2} leads to
the following result for the Gaussian noise case.

\begin{Corollary}\label{corollary1} (Gaussian Noise Case) Let $\Dbf$ be a general frame of $\Rnum^{n}$
with frame bounds $0<A\leq B<\infty$.  Let $\widetilde{\Dbf}$ be an
alternative
  dual frame of $\Dbf$ with frame bounds $0<\widetilde{A}\leq \widetilde{B}<\infty$, and let $\rho=s/b$. Suppose
  \begin{equation}
    \left(1-\sqrt{\rho B \widetilde{B}}\right)^2 \cdot \delta_{s+a} +
\rho B \widetilde{B}\cdot\delta_{b} < 1 - 2\sqrt{\rho B
\widetilde{B}}
  \end{equation}
holds for some positive integers $a$ and $b$ satisfying $0< b-a\leq
3a$. Then with probability at least $1-(1/m)$, the solution
$\hat{\fbf}$ to \eqref{generall1analysis} with $\epsilon =
\sigma\sqrt{m+2\sqrt{m\log m}}$ satisfies
  \begin{equation}
    \Vert \hat{\fbf}-\fbf \Vert_{2} \leq C_{0}\cdot\sigma\sqrt{m+2\sqrt{m\log
  m}} +
    C_{1}\cdot\f{\Vert\widetilde{\Dbf}^{*}\fbf-(\widetilde{\Dbf}^{*}\fbf)_{s}\Vert_{1}}{\sqrt{s}},
  \end{equation}
 where $C_{0}$ and $C_{1}$ are some constants and $(\widetilde{\Dbf}^{*}\fbf)_{s}$
 denotes the vector consisting the largest $s$ entries of
 $\widetilde{\Dbf}^{*}\fbf$ in magnitude.
\end{Corollary}

\subsection{An Improvement in the Case of the Canonical Dual Frame}

We also notice that when using the explicit matrix structure of the
canonical dual
$\widetilde{\Dbf}=\bar{\Dbf}=(\Dbf\Dbf^{*})^{-1}\Dbf$, the
sufficient condition can be further improved.  It seems to us that
such an improvement can not easily carry through to the general dual
frame case.

\begin{Theorem}\label{thm2}
  Let $\Dbf$ be a general frame of $\Rnum^{n}$ with frame bound $0<A\leq B<\infty$
  and $\bar{\Dbf}$ be the canonical dual frame of $\Dbf$. Let
  $\kappa=B/A$ and $\rho=s/b$ such that $\rho<1/\kappa$. Suppose
  \begin{equation}
    \label{SuffiCon} (1-\rho\kappa)^2 \cdot \delta_{s+a} +
\rho\kappa^3\cdot\delta_{b} < (1-\rho\kappa)^2 - \rho\kappa^3
  \end{equation}
holds for some positive integers $a$ and $b$ satisfying $0<b-a\leq
3a$. Then the solution $\hat{\fbf}$ to \eqref{generall1analysis} (with the canonical dual frame as the analysis operator)
satisfies
  \begin{equation}
    \label{ErrorBound} \Vert \hat{\fbf}-\fbf \Vert_{2} \leq C_{0}\cdot\epsilon +
    C_{1}\cdot\f{\Vert\bar{\Dbf}^{*}\fbf-(\bar{\Dbf}^{*}\fbf)_{s}\Vert_{1}}{\sqrt{s}},
  \end{equation}
 where $C_{0}$ and $C_{1}$ are some constants and $(\bar{\Dbf}^{*}\fbf)_{s}$
 denotes the vector consisting the largest $s$ entries of
 $\bar{\Dbf}^{*}\fbf$ in magnitude.
\end{Theorem}

\begin{proof}
In this case, \eqref{eq18} and \eqref{eq21}
respectively become
\begin{equation}
\label{eq31}\Vert\bar{\Dbf}_{T_{01}^{c}}^{*}\hbf\Vert_{2} \leq
\sum_{i\geq2}\Vert\bar{\Dbf}_{T_{i}}^{*}\hbf\Vert_{2} \leq
\sqrt{\rho}\left(\f{1}{\sqrt{{A}}}\|\hbf\|_{2}+\eta\right)
\end{equation}
and
\begin{equation}
  \label{eq32} \sqrt{1-\delta_{s+a}}\Vert\Dbf_{T_{01}}\bar{\Dbf}_{T_{01}}^{*}\hbf
  \Vert_{2} \leq 2\epsilon + \sqrt{\rho(1+\delta_{b})B}\left(\f{1}{\sqrt{{A}}}\|\hbf\|_{2}+\eta\right).
\end{equation}
We have
\begin{eqnarray}
\nonumber \Vert\hbf\Vert_{2}^{2} & = &
\Vert\Dbf\bar{\Dbf}^{*}\hbf\Vert_{2}^{2}
\stackrel{\eqref{eq10}}{\leq} B\Vert\bar{\Dbf}^{*}\hbf\Vert_{2}^{2}=
B\Vert\bar{\Dbf}_{T_{01}}^{*}\hbf\Vert_{2}^{2} +
B\Vert\bar{\Dbf}_{T_{01}^{c}}^{*}\hbf\Vert_{2}^{2}\\\nonumber & = &
B\langle(\Dbf\Dbf^{*})^{-1}\hbf,
\Dbf_{T_{01}}\bar{\Dbf}_{T_{01}}^{*}\hbf\rangle +
B\Vert\bar{\Dbf}_{T_{01}^{c}}^{*}\hbf\Vert_{2}^{2}\\\nonumber &
\stackrel{C.S.}{\leq} & B\Vert(\Dbf\Dbf^{*})^{-1}\hbf\Vert_{2}
\|\Dbf_{T_{01}}\bar{\Dbf}_{T_{01}}^{*}\hbf\|_{2}+
B\Vert\bar{\Dbf}_{T_{01}^{c}}^{*}\hbf\Vert_{2}^{2}\\ \nonumber &
\stackrel{\eqref{eq18}}{\leq} & \f{B}{A}\Vert\hbf\Vert_{2}
\Vert\Dbf_{T_{01}}\bar{\Dbf}_{T_{01}}^{*}\hbf\Vert_{2} +
B\rho\left(\f{1}{\sqrt{A}}\|\hbf\|_{2}+\eta\right)^{2}\\
\label{eq33} & = & \f{B}{A}\Vert\hbf\Vert_{2}
\Vert\Dbf_{T_{01}}\bar{\Dbf}_{T_{01}}^{*}\hbf\Vert_{2} +
\f{B\rho}{A}\|\hbf\|_2^2 + \f{2B\rho}{\sqrt{A}}
\|\hbf\|_{2}\cdot\eta + B\rho\eta^{2}.
\end{eqnarray}
Applying the fact that $uv\leq \f{cu^2}{2}+\f{v^2}{2c}$ for any
value $u$, $v$ and $c>0$ twice to \eqref{eq33}, we have
\begin{eqnarray*}
\Vert\hbf\Vert_{2}^{2} & {\leq} &
\f{B}{A}\left(\f{c_{1}\Vert\hbf\Vert_{2}^{2}}{2}+\f{\Vert\Dbf_{T_{01}}\bar{\Dbf}_{T_{01}}^{*}\hbf
\Vert_{2}^{2}}{2c_{1}}\right) + \f{B\rho}{A}\|\hbf\|_2^2 +
\f{2B\rho}{\sqrt{A}}
\|\hbf\|_{2}\cdot\eta + B\rho\eta^{2}\\
& {\leq}
&\f{B}{A}\left(\f{c_{1}\Vert\hbf\Vert_{2}^{2}}{2}+\f{\Vert\Dbf_{T_{01}}\bar{\Dbf}_{T_{01}}^{*}\hbf
\Vert_{2}^{2}}{2c_{1}}\right) + \f{B\rho}{A}\|\hbf\|_2^2 +
\f{2B\rho}{\sqrt{A}}
\left(\f{c_{2}\|\hbf\|_2^2}{2}+\f{\eta^2}{2c_2}\right) +
B\rho\eta^{2},
\end{eqnarray*}
where $c_{1}, c_{2}>0$. Let $\kappa=B/A$ and simplifying the above
equation yields
$$\left( 1-\f{c_{1}\kappa}{2}-\rho\kappa-c_{2}\rho\sqrt{\kappa B}\right)\Vert\hbf\Vert_{2}^{2} \leq \f{\kappa}{2c_{1}}\Vert\Dbf_{T_{01}}\bar{\Dbf}_{T_{01}}^{*}\hbf
\Vert_{2}^{2}+\left({\rho\sqrt{\kappa B}}/{c_{2}}+\rho
B\right)\eta^{2}.$$ Using the fact that $\sqrt{u^{2}+v^{2}}\leq u+v$
for $u,v\geq0$, we obtain
\begin{equation}\label{eq34}\Vert\hbf\Vert_{2}\sqrt{\left( 1-\f{c_{1}\kappa}{2}-\rho\kappa-c_{2}\rho\sqrt{\kappa B}\right)}
\leq
\sqrt{\f{\kappa}{2c_1}}\Vert\Dbf_{T_{01}}\bar{\Dbf}_{T_{01}}^{*}\hbf
\Vert_{2}+\eta\sqrt{\left({\rho\sqrt{\kappa B}}/{c_{2}}+\rho
B\right)}.
\end{equation}
Here we have assumed that
\begin{equation}
\label{NecCon}1-\f{c_{1}\kappa}{2}-\rho\kappa-c_{2}\rho\sqrt{\kappa B}>0.
\end{equation}
Combining \eqref{eq32} with \eqref{eq34} yields
\begin{equation}
  \label{eq35} K_{1}\Vert\hbf\Vert_{2}\leq 2\epsilon +
  K_{2}\eta,
\end{equation}
where \begin{eqnarray*} K_{1}& = &
\sqrt{\f{2c_{1}}{\kappa}(1-\delta_{s+a})\left(
1-\f{c_{1}\kappa}{2}-\rho\kappa-c_{2}\rho\sqrt{\kappa B}\right)}
-\sqrt{\rho\kappa(1+\delta_{b})},\\ K_{2} & = &
\sqrt{\f{2c_{1}}{\kappa}(1-\delta_{s+a})\left({\rho\sqrt{\kappa
B}}/{c_{2}}+\rho B\right)}+\sqrt{\rho B(1+\delta_{b})}.
\end{eqnarray*}
If $K_{1}$ is positive, then we have
\begin{equation}
  \label{eq36} \Vert\hbf\Vert_{2}\leq \f{2}{K_{1}}\cdot\epsilon +
  \f{K_{2}}{K_{1}}\cdot\eta = C_{0}\cdot\epsilon +
    C_{1}\cdot\f{\Vert\bar{\Dbf}\fbf-(\bar{\Dbf}\fbf)_{s}\Vert_{1}}{\sqrt{s}},
\end{equation}
where $C_{0}={2}/{K_{1}}$ and $C_{1}={2K_{2}}/{K_{1}}$. We now
consider how to properly choose the parameters $c_{1}, c_{2}>0$ such
that $K_{1}$ is positive and \eqref{NecCon} holds. Let $g(c_{1},
c_{2})=2c_{1}( 1-\f{c_{1}\kappa}{2}-\rho\kappa-c_{2}\rho\sqrt{\kappa B}),
\ \  c_{1}, c_{2}>0$. Note first that $g(c_{1}, c_{2})$ decreases as
$c_{2}$ increases. Thus we can take $c_{2}$ arbitrarily small, i.e.,
$c_2\rightarrow 0_{+}$, then $g(c_{1}, c_{2})$ reduces to
$g(c_{1})=2c_{1}( 1-\f{c_{1}\kappa}{2}-\rho\kappa)$. Further, $g(c_{1})$
achieves its maximum at $c_{1}^{\textrm{opt}} = (
1-\rho\kappa)/{\kappa}$. Hence, we choose $c_{1}=c_{1}^{\textrm{opt}}$
and $K_{1}>0$ is guaranteed provided that
\begin{equation}
\label{eq38}(1-\rho\kappa)^2 \cdot \delta_{s+a} +
\rho\kappa^3\cdot\delta_{b} < (1-\rho\kappa)^2 - \rho\kappa^3.
\end{equation}
To guarantee $c_1>0$ and \eqref{NecCon} holds, it is also required
that
\begin{equation}
  \label{RhoCon} \rho < \f{1}{\kappa}.
\end{equation}
This completes the proof.

\end{proof}

\noindent{\bf Remark 5:}\ The D-RIP condition can now be
$\delta_{2s} < 0.2$ in the case of Parseval frames.   Suppose $\Dbf$
is a Parseval frame and the analysis operator $\widetilde{\Dbf}$ is
its canonical dual frame, i.e., $\widetilde{\Dbf}=\Dbf$. Then
\eqref{SuffiCon} becomes, since $\kappa=1$,
\begin{equation}
\label{eq39}(1-\rho)^2 \cdot \delta_{s+a} + \rho\cdot\delta_{b} <
(1-\rho)^2 - \rho.
\end{equation}
Again, different choices of $a$ and $b$ will lead to different
conditions. For instance, let $a=s, b=4s,$ and $\rho = s/b =1/4<1$.
Then \eqref{eq39} becomes
\begin{equation}\label{eq40}
 9\delta_{2s}+4\delta_{4s}<5.
\end{equation}
which is satisfied whenever $\delta_{2s}<0.2$. Note also that
smaller $\delta_{4s}$ will lead to smaller constants in the error
bound. For example, let $c_2=1/10$ and $c_1=1-\rho-c_2\rho=29/40$,
then we have $C_0 = 29.1$ and $C_1 = 66.5$ whenever $\delta_{4s}\leq
1/4$. If $\delta_{4s}$ has a tighter restriction, i.e.,
$\delta_{4s}\leq 1/8$, then the constants become to $C_0 = 13.6$ and
$C_1 = 32.5$.

\section{Optimal-dual-based $\ell_1$-analysis and an Iterative Algorithm}\label{section4}
One of the applications of the general-dual-based $\ell_1$-analysis
and its error bound analysis is in the optimal-dual-based
$\ell_1$-analysis approach as we briefly discussed in the
introduction. Recall that our goal is to solve a constrained
optimization problem of this form\footnote{For simplicity of
notations, we replace $\tilde{\fbf}$ by $\fbf$ in this section.}:
\begin{equation}
 \label{generaldualbasedL1analysisI}
 \hat{\fbf}=\underset {\fbf\in\Rnum^{n},\  \gbf\in\Rnum^{d}} {\textrm{argmin}}
 \Vert\bar{\Dbf}^{*}\fbf + \Pbf \gbf\Vert_{1} \ \ \ \  s.t. \ \
  \Vert\ybf-\Phibf\fbf\Vert_{2} \leq \epsilon.
\end{equation}
It is well known that this problem is difficult to solve numerically
since the $\ell_1$ term involved in
\eqref{generaldualbasedL1analysisI} is nonsmooth and nonseparable.
In this section, we focus on applying the split Bregman iteration
\cite{Goldstein2009} and develop an iterative algorithm for solving
the optimal-dual-based $\ell_1$-analysis problem. Since our
derivation of this algorithm makes use of the Bregman iteration, we
include an outline of the basics of this technique in Appendix
\ref{Appendix1}.

\subsection{Optimal-dual-based $\ell_1$-analysis via Split Bregman Iteration}
The goal of the split Bregman method is to extend the utility of the
Bregman iteration to the minimization of problems involving multiple
$\ell_1$-regularization terms \cite{Goldstein2009} and
$\ell_1$-analysis \cite{CaiJ2009}. Here, we apply the split Bregman
iteration to solve the optimal-dual-based $\ell_1$-analysis problem
\eqref{generaldualbasedL1analysisI}. The basic idea is to introduce
an intermediate variable $\dbf$ such that $\dbf = \bar{\Dbf}^{*}\fbf
+ \Pbf \gbf$, and the term $\|\bar{\Dbf}^{*}\fbf + \Pbf \gbf\|_1$ in
\eqref{generaldualbasedL1analysisI} is separable and easy to
minimize.

To solve \eqref{generaldualbasedL1analysisI}, one can use the
Bregman iteration \eqref{BregmanforConstrainedIII} for the equality
constrained version of \eqref{generaldualbasedL1analysisI} with an
early stopping criterion
\begin{equation}\label{stopcriterion}
  \|\Phibf\fbf^k-\ybf\|_2 \leq \epsilon
\end{equation}
to find a good approximate solution of
\eqref{generaldualbasedL1analysisI}. This approach has already been
used and discussed in, for example, \cite{CaiJ2009},
\cite{Osher2005}, \cite{Yin2008}. The equality constrained version
of \eqref{generaldualbasedL1analysisI} is given by
\begin{equation}\label{generaldualbasedL1analysisequacon}
 \hat{\fbf}=\underset {\fbf\in\Rnum^{n},\  \gbf\in\Rnum^{d}} {\textrm{argmin}}
 \Vert\bar{\Dbf}^{*}\fbf + \Pbf \gbf\Vert_{1} \ \ \ \  s.t. \ \
  \Phibf\fbf=\ybf.
\end{equation}
Apply the Bregman iteration \eqref{BregmanforConstrainedIII} to the
constrained minimization problem
\eqref{generaldualbasedL1analysisequacon}, we obtain
\begin{equation}\label{SplitBregmanforConstrained}
  \left\{\begin{array}{l}  (\fbf^{k+1}, \gbf^{k+1}) =  {\textrm{argmin}_{\fbf, \gbf}}
  \|\bar{\Dbf}^{*}\fbf + \Pbf \gbf\|_1 + \frac{\mu}{2}\|\Phibf\fbf - \ybf + \cbf^k\|_2^2, \\ \cbf^{k+1} = \cbf^k + (\Phibf\fbf^{k+1} - \ybf),\end{array}  \right.
\end{equation}
for $k=0, 1, \ldots,$ starting with $\cbf^0 = \zebf$, $\gbf^0 =
\zebf$, and  $\fbf^0 = \zebf$. In the first step, we have to solve a
subproblem of this form
\begin{equation}\label{subproblem}
  \underset {\fbf,\ \gbf} {\textrm{min}}
  \|\bar{\Dbf}^{*}\fbf + \Pbf \gbf\|_1 + \frac{\mu}{2}\|\Phibf\fbf - \ybf +
  \cbf^k\|_2^2.
\end{equation}
This problem is equivalent to
\begin{equation}
  \label{GeneralL1regularizationEQ} \underset {\fbf,\  \gbf, \ \dbf} {\textrm{min}} \ \
  \|\dbf\|_1+ \frac{\mu}{2}\|\Phibf\fbf - \ybf +
  \cbf^k\|_2^2  \ \ s.t. \ \ \dbf = \bar{\Dbf}^{*}\fbf + \Pbf \gbf.
\end{equation}
Again, apply the Bregman iteration \eqref{BregmanforConstrainedIII}
to \eqref{GeneralL1regularizationEQ}, we have the following
two-phase algorithm for solving the subproblem \eqref{subproblem}
\begin{equation}\label{SBI_I}
  \left\{\begin{array}{l}
  (\fbf^{k+1}, \dbf^{k+1}, \gbf^{k+1})   =  {\textrm{argmin}_{\fbf, \dbf, \gbf}}
  \|\dbf\|_1+ \frac{\mu}{2}\|\Phibf\fbf - \ybf +
  \cbf^k\|_2^2  + \frac{\lambda}{2}\|\bar{\Dbf}^{*}\fbf + \Pbf \gbf-\dbf + \bbf^k\|_2^2, \\ \bbf^{k+1}  = \bbf^k + (\bar{\Dbf}^{*}\fbf^{k+1} + \Pbf \gbf^{k+1} - \dbf^{k+1}).
  \end{array}  \right.
\end{equation}

Since we have split the $\ell_1$ and $\ell_2$ components of the
subproblem involved in \eqref{SBI_I}, we can perform this
minimization efficiently by iteratively minimizing with respect to
$\fbf$, $\dbf$, and $\gbf$ separately. Thus we arrive at the following
three steps:
\begin{align} \label{step1}
  \text{Step 1}:  \fbf^{k+1} & = {\textrm{argmin}_{\fbf}}
  \frac{\mu}{2}\|\Phibf\fbf - \ybf +
  \cbf^k\|_2^2 + \frac{\lambda}{2}\|\bar{\Dbf}^{*}\fbf + \Pbf \gbf^k-\dbf^k +
  \bbf^k\|_2^2,\\ \label{step2}
   \text{Step 2}:  \dbf^{k+1} & = {\textrm{argmin}_{\dbf}}
  \|\dbf\|_1 + \frac{\lambda}{2}\|\dbf - \bar{\Dbf}^{*}\fbf^{k+1} - \Pbf \gbf^k -
  \bbf^k\|_2^2,\\ \label{step3}
  \text{Step 3}:  \gbf^{k+1} & = {\textrm{argmin}_{\gbf}}
  \frac{\lambda}{2}\|\Pbf \gbf + \bar{\Dbf}^{*}\fbf^{k+1} - \dbf^{k+1} +
  \bbf^k\|_2^2.
\end{align}

In Step $1$, because we have decoupled $\fbf$ from the $\ell_1$
portion of the problem, the optimization problem is now
differentiable. The optimality conditions to \eqref{step1} yield
\begin{equation}\label{steponeoptcond}
\mu \Phibf^*(\Phibf\fbf - \ybf +
  \cbf^k) + \lambda \bar{\Dbf} (\bar{\Dbf}^{*}\fbf + \Pbf \gbf^k-\dbf^k +
  \bbf^k) = 0.
\end{equation}
Thus we can compute
\begin{equation} \label{updatef}
  \fbf^{k+1} = (\mu\Phibf^*\Phibf + \lambda
  \bar{\Dbf}\bar{\Dbf}^{*})^{-1}[\mu \Phibf^*(\ybf -
  \cbf^k) + \lambda \bar{\Dbf} (\dbf^k - \Pbf \gbf^k -
  \bbf^k)].
\end{equation}

In Step $2$, there is no coupling between elements of $\dbf$. This
problem can be solved by a simple soft shrinkage, i.e.,
\begin{equation}\label{updated}
  \dbf^{k+1} = \text{shrink} (\bar{\Dbf}^{*}\fbf^{k+1}+\Pbf \gbf^k +
  \bbf^k, 1/\lambda),
\end{equation}
where the soft shrinkage operator is defined as
\begin{equation*}
\text{shrink}(\wbf_i, 1/\lambda) = \text{sign}(\wbf_i)\cdot
\text{max}(|\wbf_i|-1/\lambda, 0).
\end{equation*}

In Step $3$, the optimality conditions to \eqref{step3} lead to
\begin{equation} \label{optcondtionstep3}
\lambda \Pbf (\Pbf \gbf + \bar{\Dbf}^{*}\fbf^{k+1} - \dbf^{k+1} +
  \bbf^k) = 0.
\end{equation}
Since only $\Pbf\gbf^k$ is involved in the update of  $\fbf^k$,
$\dbf^k$, and $\bbf^k$, it is enough to derive an updating formula
for $\Pbf\gbf^k$
\begin{equation}
  \Pbf\gbf^{k+1} = \Pbf(\dbf^{k+1} - \bar{\Dbf}^{*}\fbf^{k+1} -
  \bbf^k).
\end{equation}

Therefore, we obtain the unconstrained split Bregman algorithm for
solving the subproblem \eqref{subproblem} as follows:
\begin{equation}\label{GSBI}
  \left\{\begin{array}{l} \text{for n = 1 to N}
  \\ \ \ \   \fbf^{k+1} = (\mu\Phibf^*\Phibf + \lambda
  \bar{\Dbf}\bar{\Dbf}^{*})^{-1}[\mu \Phibf^*(\ybf -
  \cbf^k) + \lambda \bar{\Dbf} (\dbf^{new} - \Pbf \gbf^{new} -
  \bbf^{k})], \\ \ \ \ \dbf^{k+1} = \text{shrink} (\bar{\Dbf}^{*}\fbf^{new}+\Pbf \gbf^{new} +
  \bbf^k, 1/\lambda), \\ \ \ \ \Pbf\gbf^{k+1} = \Pbf(\dbf^{new} - \bar{\Dbf}^{*}\fbf^{new} -
  \bbf^k), \\
  \text{end} \\
  \bbf^{k+1}  = \bbf^k + (\bar{\Dbf}^{*}\fbf^{k+1} + \Pbf \gbf^{k+1} -
  \dbf^{k+1}),
  \end{array}  \right.
\end{equation}
where $(\cdot)^{new}$ denotes either $(\cdot)^{k+1}$ if it is
available or $(\cdot)^{k}$ otherwise.

Ideally, we need to run infinite iterations ($N \rightarrow \infty$)
to obtain a convergent solution for the subproblem involved in
\eqref{SBI_I}. However, as pointed out in \cite{Goldstein2009}, it
is not desirable to solve this subproblem to full convergence.
Intuitively, the reason for this is that if the error in our
solution for this subproblem is small compared to $\|\bbf^k -
\bbf^\sharp\|_2$, where $\bbf^\sharp$ is the ``true $\bbf$'', then
this extra precision will be ``wasted'' when the Bregman parameter
is updated. In fact, it was found empirically in
\cite{Goldstein2009} that for many applications optimal efficiency
is obtained when only one iteration of the inner loop is performed
(i.e., $N = 1$ in \eqref{GSBI}). When $N=1$, the unconstrained split
Bregman iteration \eqref{GSBI} reduces to

\begin{equation}\label{SBINequal1}
  \left\{\begin{array}{l}
  \fbf^{k+1} = (\mu\Phibf^*\Phibf + \lambda
  \bar{\Dbf}\bar{\Dbf}^{*})^{-1}[\mu \Phibf^*(\ybf -
  \cbf^k) + \lambda \bar{\Dbf} (\dbf^k - \Pbf \gbf^k -
  \bbf^k)], \\ \dbf^{k+1} = \text{shrink} (\bar{\Dbf}^{*}\fbf^{k+1}+\Pbf \gbf^k +
  \bbf^k, 1/\lambda),  \\
   \Pbf\gbf^{k+1} = \Pbf(\dbf^{k+1} - \bar{\Dbf}^{*}\fbf^{k+1} -
  \bbf^k),\\ \bbf^{k+1}  = \bbf^k + (\bar{\Dbf}^{*}\fbf^{k+1} + \Pbf \gbf^{k+1} -
  \dbf^{k+1}).
  \end{array}  \right.
\end{equation}

Combining this inner solver with the outer iteration
\eqref{SplitBregmanforConstrained}, we obtain the constrained split
Bregman method for \eqref{generaldualbasedL1analysisequacon} as
follows:
\begin{equation}\label{SBCON}
  \left\{\begin{array}{l} \text{for n = 1 to} \ \ nInner
  \\ \ \ \   \fbf^{k+1} = (\mu\Phibf^*\Phibf + \lambda
  \bar{\Dbf}\bar{\Dbf}^{*})^{-1}[\mu \Phibf^*(\ybf -
  \cbf^k) + \lambda \bar{\Dbf} (\dbf^{new} - \Pbf \gbf^{new} -
  \bbf^{new})], \\ \ \ \ \dbf^{k+1} = \text{shrink} (\bar{\Dbf}^{*}\fbf^{new}+\Pbf \gbf^{new} +
  \bbf^{new}, 1/\lambda), \\ \ \ \ \Pbf\gbf^{k+1} = \Pbf(\dbf^{new} - \bar{\Dbf}^{*}\fbf^{new} -
  \bbf^{new}), \\ \ \ \ \bbf^{k+1}  = \bbf^{new} + (\bar{\Dbf}^{*}\fbf^{new} + \Pbf \gbf^{new} -
  \dbf^{new}),\\
  \text{end} \\
  \cbf^{k+1} = \cbf^k + (\Phibf\fbf^{k+1} - \ybf),
  \end{array}  \right.
\end{equation}
where $nInner$ denotes the number of inner loops. A formal statement
of the split Bregman iteration for optimal-dual-based
$\ell_1$-analysis is given in Algorithm $1$ in which $\fbf$ denotes
the recovered signal and $\dbf$ is the recovered coefficient vector.
\begin{algorithm} \label{algorithm1}
\caption{Split Bregman Iteration for optimal-dual-based
$\ell_1$-analysis}
  \textbf{Initialization:} {$\fbf^0=\zebf$, $\dbf^0=\bbf^0=\Pbf\gbf^0=\zebf$, $\cbf^0=\zebf$, $\mu>0, \lambda>0, nOuter, nInner,
  tol$}\;
  \While{$k < nOuter$ and $\|\Phibf \fbf^k - \ybf\|_2 > tol$}
  { \For {$n=1:nInner$}{
  $\fbf^{k+1} = (\mu\Phibf^*\Phibf + \lambda
  \bar{\Dbf}\bar{\Dbf}^{*})^{-1}[\mu \Phibf^*(\ybf -
  \cbf^k) + \lambda \bar{\Dbf} (\dbf^{new} - \Pbf \gbf^{new} -
  \bbf^{new})]$\;
  $\dbf^{k+1} = \text{shrink} (\bar{\Dbf}^{*}\fbf^{new}+\Pbf \gbf^{new} +
  \bbf^{new}, 1/\lambda)$\;
  $\Pbf\gbf^{k+1} = \Pbf(\dbf^{new} - \bar{\Dbf}^{*}\fbf^{new} -
  \bbf^{new})$\;
  $\bbf^{k+1}  = \bbf^{new} + (\bar{\Dbf}^{*}\fbf^{new} + \Pbf \gbf^{new} -
  \dbf^{new})$\;
  }
  $\cbf^{k+1} = \cbf^k + (\Phibf\fbf^{k+1} - \ybf)$\;
  Increase $k$\;
  }
\end{algorithm}

\noindent {\bf Remark 6:} \ If $\Dbf$ is a Parseval frame and
$\Pbf\gbf \equiv \zebf$, then Algorithm \ref{algorithm1} reduces to
the split Bregman iteration for the standard $\ell_1$-analysis
approach as discussed in \cite{CaiJ2009}.

\subsection{Computational Complexity Analysis}

We discuss briefly the computational complexity of Algorithm
\ref{algorithm1} in this subsection. For simplicity of the
discussion, we assume that $\Dbf$ is a Parseval frame. This stems
from the fact that Parseval frames are often favored in practical
situations. Let $\Qbf \equiv (\mu\Phibf^*\Phibf + \lambda
\Ibf_n)^{-1}$. Define $\Cset_{\Phibf}$, $\Cset_{\Dbf}$, and
$\Cset_{\Qbf}$ to be the complexity of applying $\Phibf$ or
$\Phibf^*$, $\Dbf$ or $\Dbf^*$, and $\Qbf$ to a vector,
respectively. The complexity of the first step in the inner loop is
$\Cset_{\Qbf}+\Cset_{\Phibf}+\Cset_{\Dbf}$. Here the cost of vector
operations is omitted since most of the work is in matrix-vector
products for large-scale problems. Steps $2$ and $3$ in the inner
loop require the application of $\Dbf$ or $\Dbf^*$ one and two times
respectively (the matrix-vector multiplication $\Dbf^*\fbf^{new}$
from the $\dbf^k$ update can be reused). The last step in the inner
loop only involves vector operations. Hence, the total complexity of
a single inner loop is $\Cset_{\Qbf}+\Cset_{\Phibf}+4\Cset_{\Dbf}$.
Furthermore, the total cost for an outer iteration is $nInner \times
(\Cset_{\Qbf}+\Cset_{\Phibf}+4\Cset_{\Dbf})+\Cset_{\Phibf}$.

The calculations above are in some sense overly pessimistic. In
compressed sensing applications, one often encounters a matrix
$\Phibf$ as a submatrix of a unitary transform, which admits for
easy storage and fast multiplication. Important examples include the
partial Discrete Fourier Transform (DFT). By applying the matrix
inversion lemma, it is not hard to show that $\Qbf =
\frac{1}{\lambda}\left(\Ibf_n-\frac{\mu}{\lambda+\mu}\Phibf^*\Phibf\right)$.
Thus computing $\fbf^{k+1}$ in the inner loop is cheap since no
matrix inversion is required. In this case, the total costs for a
single inner loop and an outer iteration become
$2\Cset_{\Phibf}+4\Cset_{\Dbf}$ and $nInner \times
(2\Cset_{\Phibf}+4\Cset_{\Dbf})+ \Cset_{\Phibf}$, respectively.
Another important example in compressed sensing is when $\Phibf$ is
a random matrix. It is well known that in this case the eigenvalues
of $\Phibf^*\Phibf$ are well clustered. Then applying $\Qbf =
(\mu\Phibf^*\Phibf + \lambda \Ibf_n)^{-1}$ to a vector can be
computed very efficiently via a few conjugate gradient (CG) steps
\cite{Becker2011}.

As discussed earlier, if $\Pbf\gbf\equiv\zebf$, then Algorithm
\ref{algorithm1} reduces to the split Bregman iteration for the
standard $\ell_1$-analysis approach. Evidently, the corresponding
complexity for a single inner loop reduces to
$\Cset_{\Qbf}+\Cset_{\Phibf}+2\Cset_{\Dbf}$ (step $3$ disappears in
this case). This means that the cost for an inner loop decreases by
$2\Cset_{\Dbf}$. It should be pointed out that, in practical
applications, there is often a fast algorithm for applying $\Dbf$
and $\Dbf^*$, e.g., a fast wavelet transform or a fast short-time
Fourier transform \cite{Mallat1998}, which makes applying of $\Dbf$
and $\Dbf^*$ low-cost.

\section{Numerical Results} \label{section5}
In this section, we present some numerical experiments illustrating
the effectiveness of signal recovery via the optimal-dual-based
$\ell_1$-analysis approach. Our results confirm that when signals
are sparse with respect to redundant frames, the optimal-dual-based
$\ell_1$-analysis approach often achieves better recovery
performance than the standard $\ell_1$-analysis method, and that
this recovery is robust with respect to noise.

In these experiments, we use two types of frames: Gabor frames and a
concatenation of the coordinate and Fourier bases. The
optimal-dual-based $\ell_1$-analysis problems are solved by
Algorithm $1$, while the $\ell_1$-analysis problems are by Algorithm
$1$ with $\Pbf\gbf\equiv \zebf$. The sensing matrix $\Phibf$ is a
Gaussian matrix with $m=32, n=128$. The noise $\zbf$ has a white
Gaussian distribution with zero-mean and second-order moments
$\sigma^2\Ibf_m$.

\noindent {\bf Example 1: Gabor Frames.}\ \  Recall that for a window function
$g$ and positive time-frequency shift parameters $\alpha$ and
$\beta$, the Gabor frame is given by
\begin{equation}
  \{g_{_{l,k}}(t) = g(t-k\alpha)e^{2\pi il\beta t}\}_{l,k}.
\end{equation}
For many imaging systems such as radar and sonar, the received signal $f$ often has
the form
\begin{equation}
  f(t) = \sum_{k=1}^s a_k g(t-t_k)e^{i\omega_k t}.
\end{equation}
Evidently, if $s$ is small, $f$ is sparse with respect to some Gabor
frame. In this experiment, we construct a Gabor dictionary with
Gaussian windows, oversampled by a factor of $20$ so that
$d=20\times n=2560$. The tested signal $\fbf$ is sparse with respect
to the constructed Gabor frame with sparsity $s =
\text{ceil}(0.2\times m) = 7$. The positions of the nonzero entries
of the coefficient vector $\xbf$ are selected uniformly at random,
and each nonzero value is sampled from standard Gaussian
distribution. We set $\lambda = \mu =1$, $tol = 10^{-6}$, and
$nOuter = 200$ in Algorithm \ref{algorithm1}.

 \begin{figure}
 \begin{center}
 \begin{tabular}{cc}
 \includegraphics[height=6cm]
 {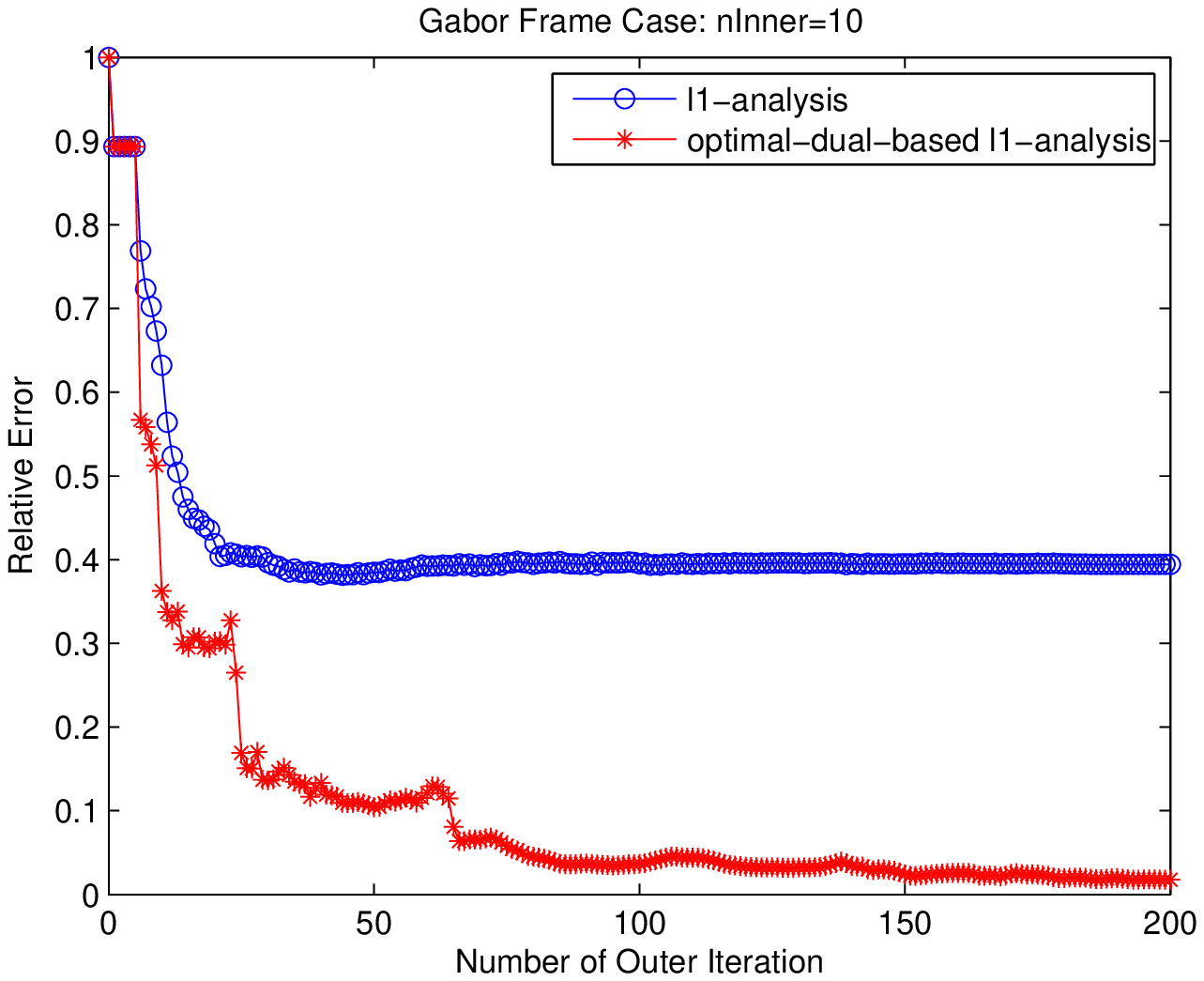}&
 \includegraphics[height=6cm]
 {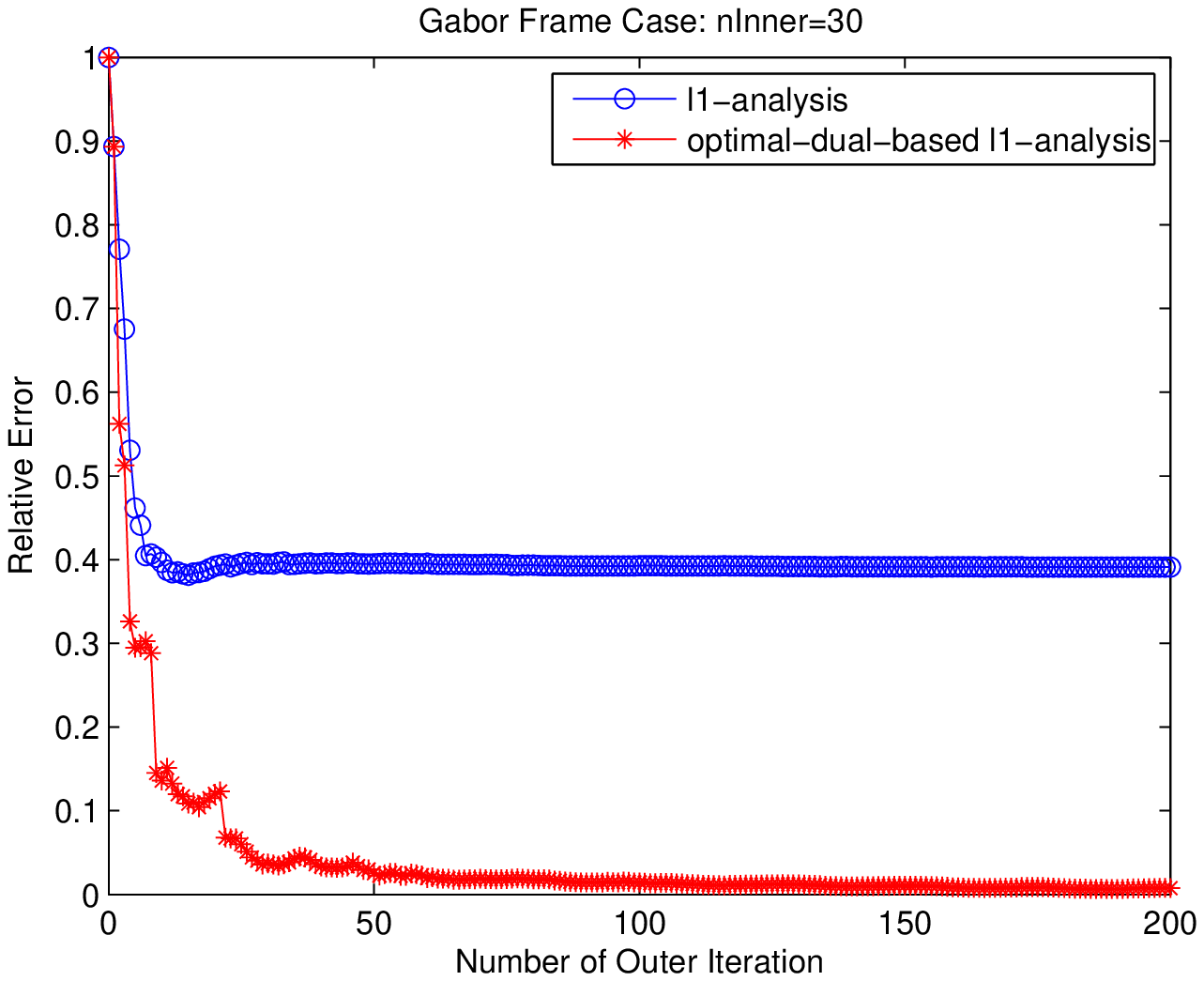}
 \end{tabular}
 \end{center}
 \caption{\small{Relative error vs. outer iteration number (without noise). The relative error at iteration $k$ is defined as ${\|\fbf-\fbf^k\|_2}/{\|\fbf\|_2}$,  where $\fbf^k$ is the approximation at iteration $k$ and $\fbf$ is the true solution. The optimal-dual-based $\ell_1$-analysis problems are solved by Algorithm $1$, while the
$\ell_1$-analysis problems are by Algorithm $1$ with $\Pbf\gbf\equiv\zebf$. Left: Results for $nInner =10$. Right: Results for $nInner =30$. } }
 \label{Fig1}
 \end{figure}

Figure \ref{Fig1} shows the relative error vs. outer iteration
number for both approaches in noiseless case\footnote{The problem of
the same setting is tested many times with randomly generated
examples (as detailed). These test results are similar to that of
Figure \ref{Fig1}. To facilitate the explanation, we only show the
result for one random instance.}. It is not hard to see that the
optimal-dual-based $\ell_1$-analysis approach is more effective than
the standard $\ell_1$-analysis approach. This is because the
optimization of the former is not only over the signal space but
also over all dual frames of $\Dbf$. In other words, there exists
some optimal dual frame ${\tilde {\bf D}}_\text{o}$ which produces
sparser coefficients than the canonical dual frame does for the
tested signal. Since ${\tilde {\bf D}}_\text{o}$ is also a dual
frame, it then follows from \eqref{errorboundforgenernalL1analysis}
that a better recovery performance can be achieved by the
optimal-dual-based $\ell_1$-analysis approach.

The convergence performance of Algorithm \ref{algorithm1} can also
be observed in Figure \ref{Fig1}. The proposed algorithm converges
quickly for the first several iterations, but then slows down as the
true solution is near. It is also evident that as $nInner$
increases, the proposed algorithm requires less outer iterations to
converge. This is because the subproblem involved in
\eqref{SplitBregmanforConstrained} is solved more accurately as
$nInner$ increases, the need for outer Bregman updates is naturally
less in order to reach the steady state. It is worth noting that as
$nInner$ increases, the corresponding complexity for an outer
iteration also increases.

Our next simulation is to show the robustness of the
optimal-dual-based $\ell_1$-analysis with respect to noise in the
measurements. Figure \ref{Fig3} shows the recovery error as a
function of the noise level. As expected, the relation is linear. We
also see that the constant $C_0$ in Theorem \ref{thm1} for the
optimal-dual-based $\ell_1$-analysis is larger than that for the
standard $\ell_1$-analysis.  But the overall performance of the
optimal-dual-based method is still much better.

 \begin{figure}
 \begin{center}
 \begin{tabular}{c}
 \includegraphics[height=8cm]
 {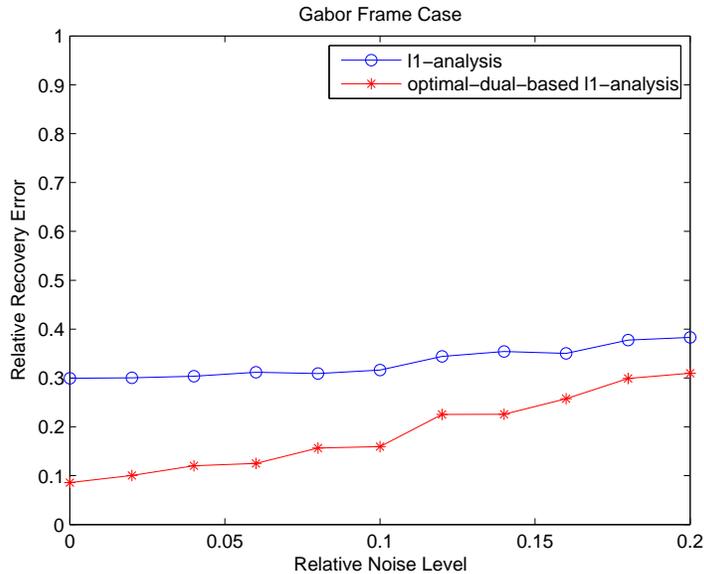}
 \end{tabular}
 \end{center}
 \caption{\small{Relative recovery error vs. relative noise level, averaged over $5$ trials. The relative recovery error is defined
 as ${\|\fbf-\hat{\fbf}\|_2}/{\|\fbf\|_2}$ and the relative noise level is defined as
 $\sqrt{m}\sigma/\|\Phibf\fbf\|_2$. The sparsity level is $s = \text{ceil}(0.2 \times m)=7$. Set $\lambda = \mu =1$, $tol
 = 10^{-6}$, $nInner = 30$, and $nOuter = 200$ in Algorithm
 \ref{algorithm1}.}}
 \label{Fig3}
 \end{figure}

We also test the performance of the optimal-dual-based
$\ell_1$-analysis with respect to the sparsity level of the
coefficient vector $\xbf$. Figure \ref{Fig4} shows that the
optimal-dual-based $\ell_1$-analysis outperforms the standard
$\ell_1$-analysis at different sparsity levels. The plot also shows
that the performance curve of the optimal-dual-based
$\ell_1$-analysis exhibits a threshold effect. When $\varrho\equiv
s/m \leq 0.2$, the optimal-dual-based $\ell_1$-analysis recovers the
signal accurately. When $\varrho \geq 0.2$, the performance degrades
as $\varrho$ increases.

 \begin{figure}
 \begin{center}
 \begin{tabular}{c}
 \includegraphics[height=8cm]
 {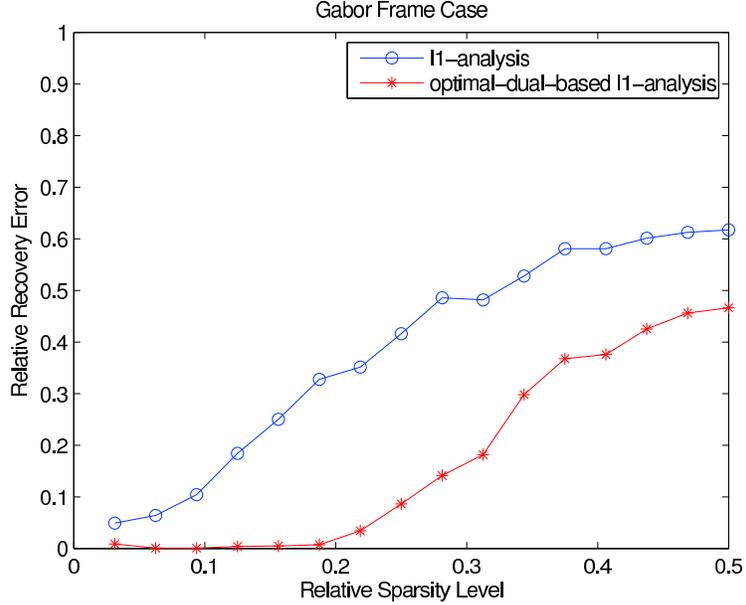}
 \end{tabular}
 \end{center}
 \caption{\small{Relative recovery error vs. relative sparsity level of $\xbf$, averaged over $100$ trials. The relative recovery error is defined
 as ${\|\fbf-\hat{\fbf}\|_2}/{\|\fbf\|_2}$ and the relative sparsity level $\varrho$ is defined
 as $\varrho = s/m$. No noise $\sigma^2=0$. The parameters for Algorithms \ref{algorithm1} are the same as in Figure \ref{Fig3}. }}
 \label{Fig4}
 \end{figure}

\noindent {\bf Example 2: Concatenations.}\ \ In many applications,
signals of interest are sparse over several orthonormal bases (or
frames), it is natural to use a dictionary $\Dbf$ consisting of a
concatenation of these bases (or frames). In this experiment, we
consider a dictionary consisting of the coordinate and Fourier
bases, i.e., $\Dbf = [\Ibf, \Fbf]$. The tested signal $\fbf$ is a
linear combination of spikes and sinusoids with sparsity $s =
\text{ceil}(0.2\times m) = 7$. The positions of the nonzero entries
of $\xbf$ are selected uniformly at random, and each nonzero value
is sampled from standard Gaussian distribution. We set $\lambda =
\mu =1$, $tol = 10^{-12}$ and $nOuter = 100$ in Algorithm
\ref{algorithm1}.

Figure \ref{Fig2} shows that the optimal-dual-based
$\ell_1$-analysis approach achieves much better recovery performance
than that of the standard $\ell_1$-analysis approach.  The latter
fundamentally fails with a relative error at about $80\%$. Such a
failure is not surprising since $\Dbf^* \fbf$ in this case is not at
all sparse.  This is due to the fact that, in this very example, the
component that is sparse in one basis is not at all in the other.

Figures \ref{Fig5} and \ref{Fig6} show the performance of the
optimal-dual-based $\ell_1$-analysis with respect to the noise level
and the sparsity level for the $\Ibf+\Fbf$ case, respectively. The
results are similar to that for the Gabor frame case. We also see
that the standard $\ell_1$-analysis fails at all noise levels and
sparsity levels in this case.

 \begin{figure}
 \begin{center}
 \begin{tabular}{cc}
 \includegraphics[height=6cm]
 {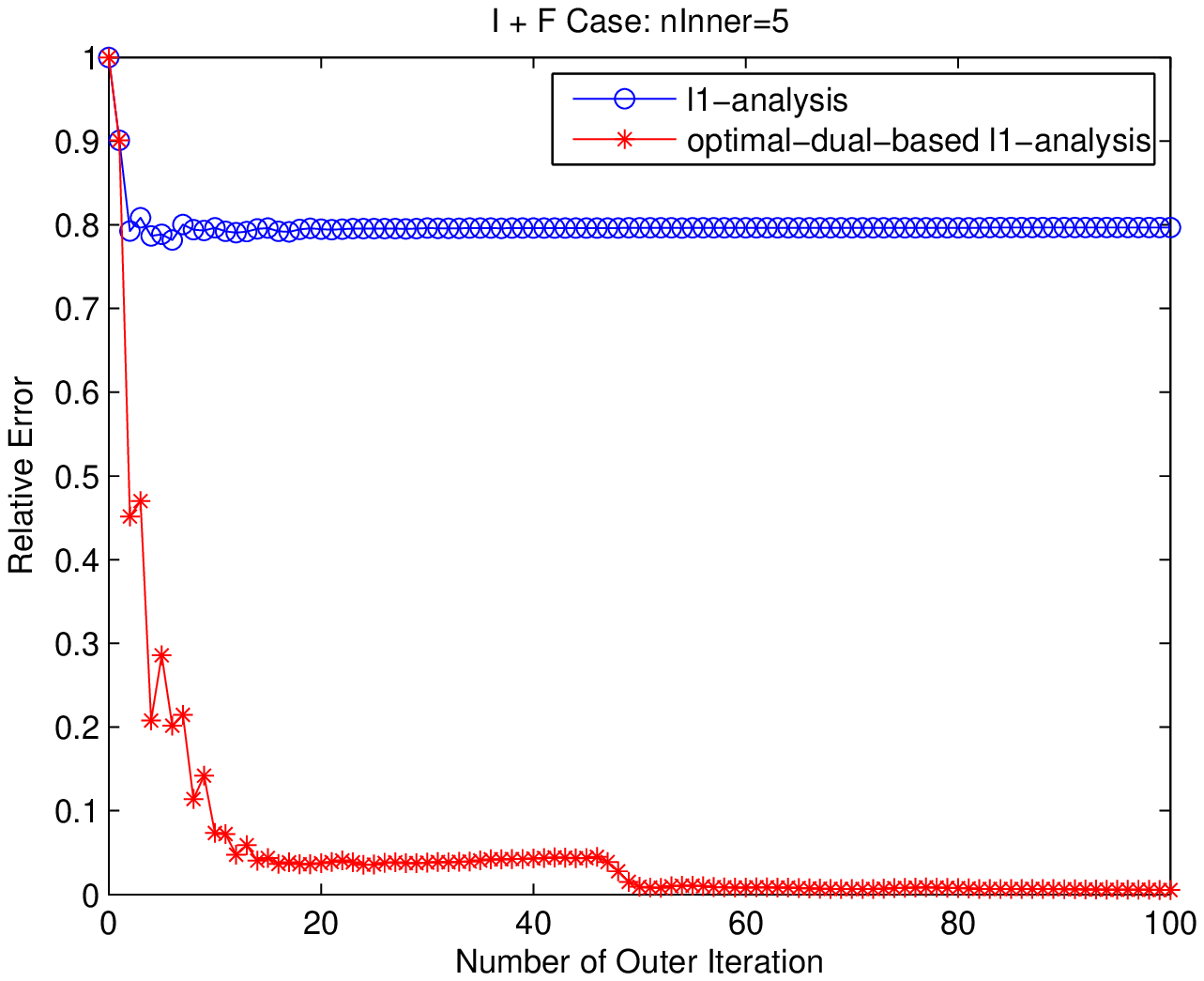}&
 \includegraphics[height=6cm]
 {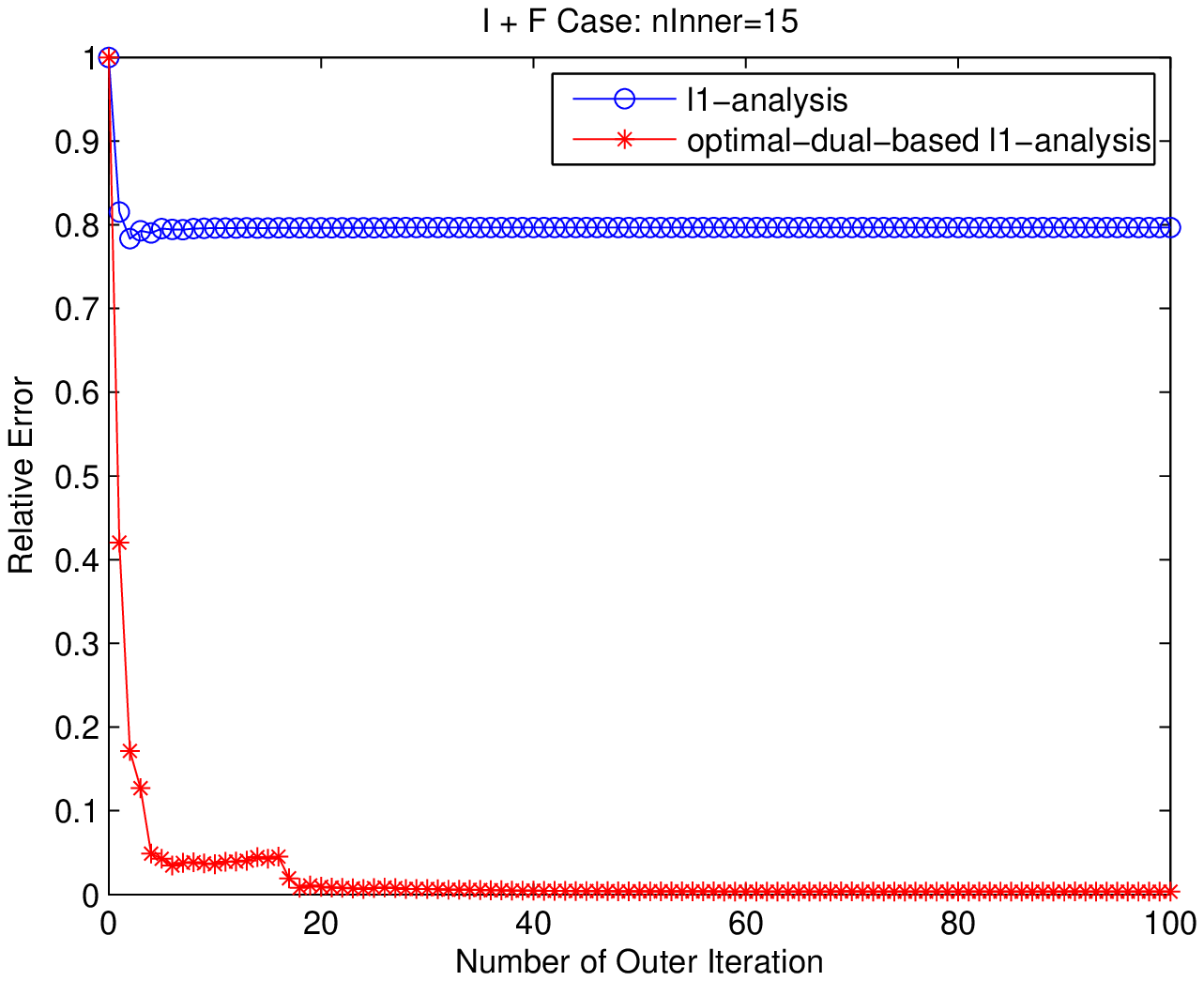}
 \end{tabular}
 \end{center}
 \caption{\small{Relative error vs. outer iteration number (without noise). The relative error at iteration $k$ is defined as ${\|\fbf-\fbf^k\|_2}/{\|\fbf\|_2}$, where $\fbf^k$ is the approximation at iteration $k$ and $\fbf$ is the true solution. The optimal-dual-based
$\ell_1$-analysis problems are solved by Algorithm $1$, while the
$\ell_1$-analysis problems are by Algorithm $1$ with $\Pbf\gbf\equiv
\zebf$. Left: Results for $nInner =5$. Right: Results for $nInner
=15$.}}
 \label{Fig2}
 \end{figure}

 \begin{figure}
 \begin{center}
 \begin{tabular}{c}
 \includegraphics[height=8cm]
 {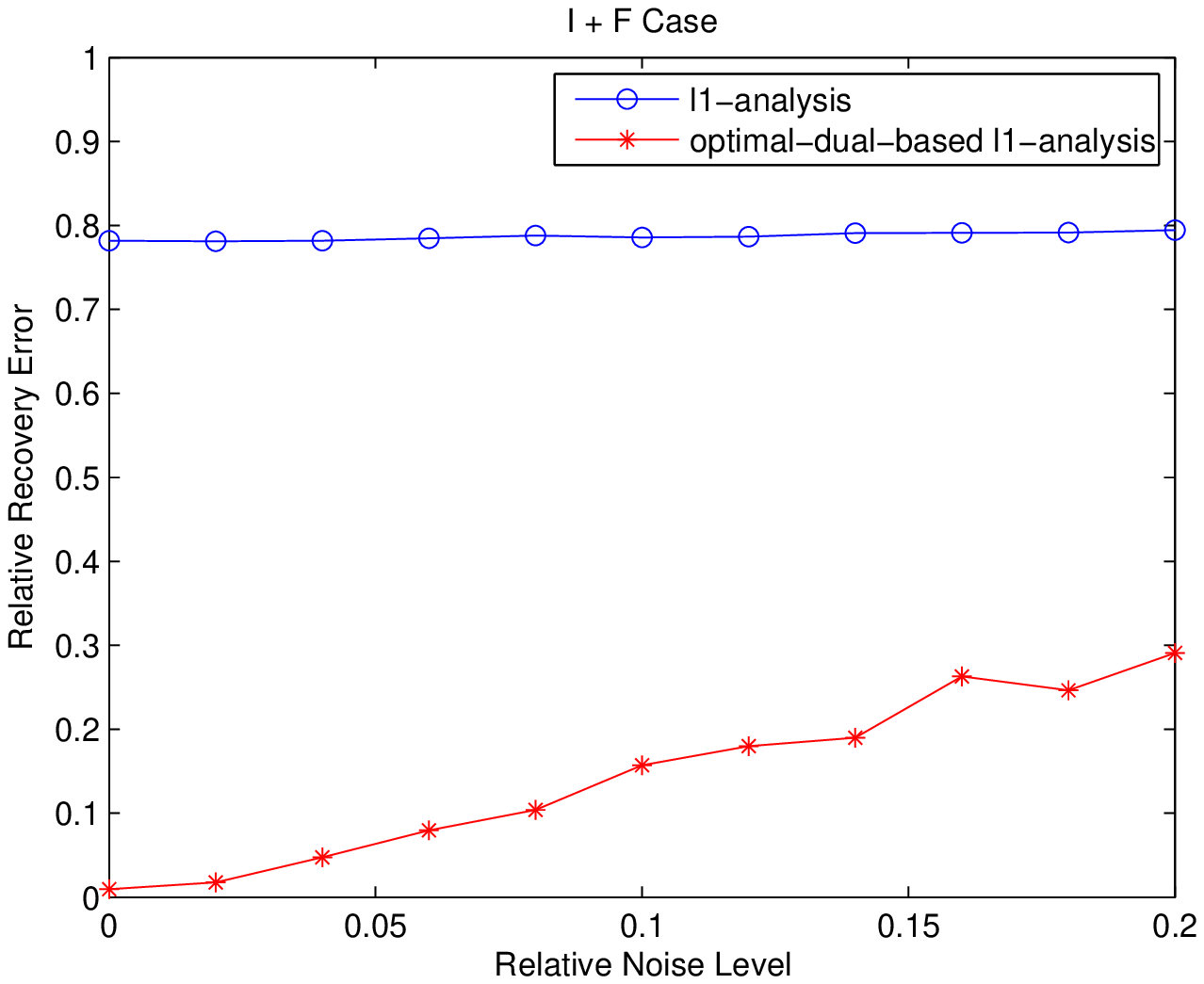}
 \end{tabular}
 \end{center}
 \caption{\small{Relative recovery error vs. relative noise level, averaged over $5$ trials. The relative recovery error is defined
 as ${\|\fbf-\hat{\fbf}\|_2}/{\|\fbf\|_2}$ and the relative noise level is defined as
 $\sqrt{m}\sigma/\|\Phibf\fbf\|_2$. The sparsity level is $s = \text{ceil}(0.2 \times m)=7$. Set $\lambda = \mu =1$, $tol
 = 10^{-12}$, $nInner = 15$, and $nOuter = 100$ in Algorithm
 \ref{algorithm1}.}}
 \label{Fig5}
 \end{figure}

 \begin{figure}
 \begin{center}
 \begin{tabular}{c}
 \includegraphics[height=8cm]
 {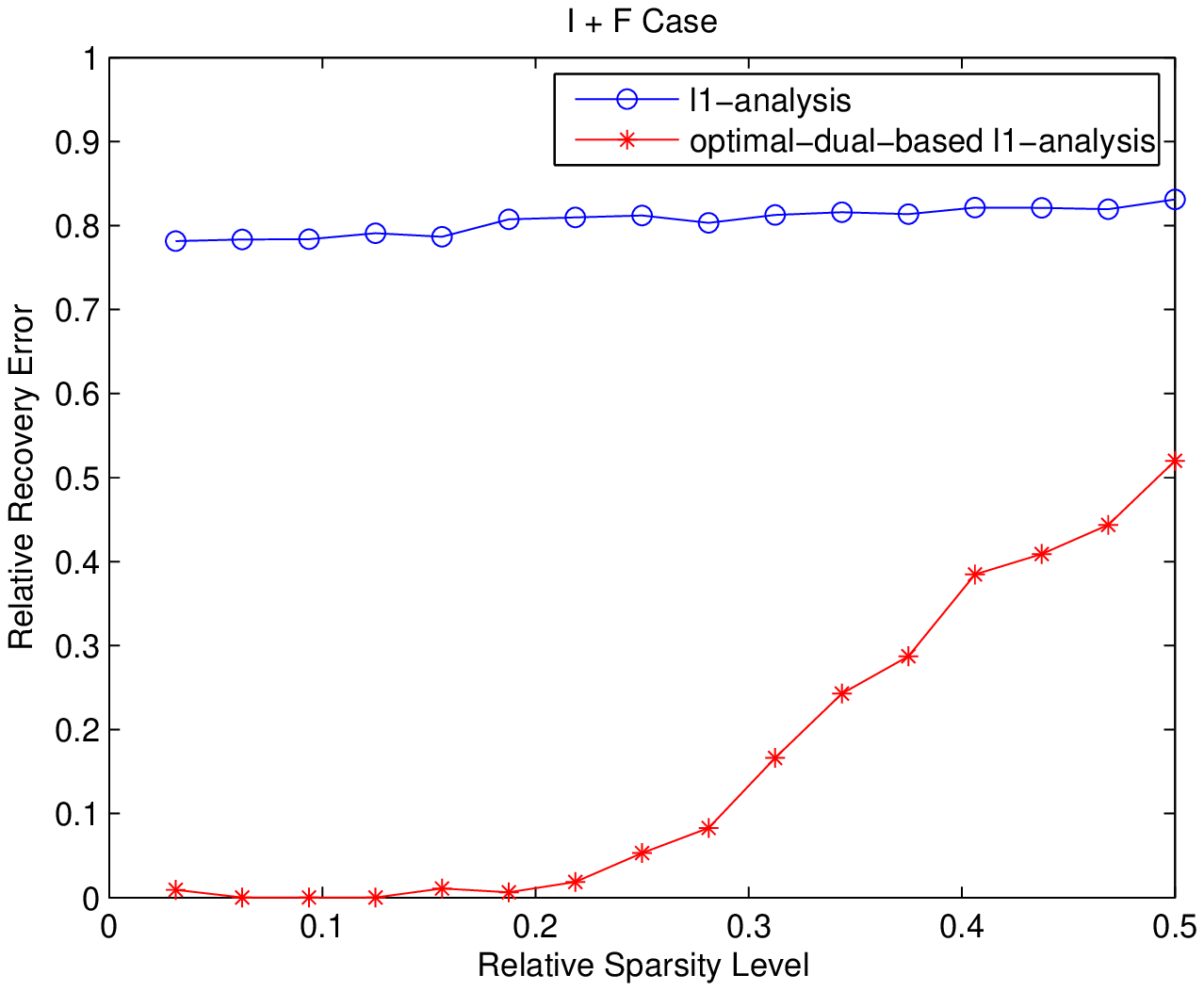}
 \end{tabular}
 \end{center}
 \caption{\small{Relative recovery error vs. relative sparsity level of $\xbf$, averaged over $100$ trials. The relative recovery error is defined
 as ${\|\fbf-\hat{\fbf}\|_2}/{\|\fbf\|_2}$ and the relative sparsity level $\varrho$ is defined
 as $\varrho = s/m$. No noise $\sigma^2=0$. The parameters for Algorithms \ref{algorithm1} are the same as in Figure \ref{Fig5}.}}
 \label{Fig6}
 \end{figure}

\section{Conclusions}\label{section6}
We extend the $\ell_1$-analysis approach to a more general case in
which the analysis operator can be any dual frame of $\Dbf$. We call
it the general-dual-based approach.  Error performance bound is
established. Improved sufficient signal recovery conditions are
provided. To demonstrate the effectiveness of the general-dual-based
approach, we also propose an optimal-dual-based $\ell_1$-analysis
approach to recover the signal directly. The optimization of this
method is not only over the signal space but also over all dual
frames of $\Dbf$. We have seen that when signals are sparse with
respect to frames that are redundant and coherent, this
optimal-dual-based approach often achieves better recovery
performance than that of the standard $\ell_1$-analysis.  By
applying the split Bregman iteration, we develop an iterative
algorithm for solving the optimal-dual-based $\ell_1$-analysis
problem. The proposed algorithm is very fast when proper parameter
values are used and easy to code. Our ongoing work includes the
performance analysis of the $\ell_1$-synthesis approach by virtue of
the principle of the optimal-dual-based $\ell_1$-analysis approach
we proposed, and further refinements of the algorithm.

\appendix
\section{The Basics of the Bregman Iteration} \label{Appendix1}
The Bregman iteration is a technique that originated in functional
analysis for finding extrema of convex functionals
\cite{Bregman1967}. The Bregman iteration was first introduced to
image processing in \cite{Osher2005}, where it was applied to total
variation (TV) denoising. Then, in \cite{Cai2009a}, \cite{Yin2008},
it was shown to be remarkably successful for $\ell_1$ minimization
problems in compressed sensing. Here we briefly review this
technique. More details about the Bregman iteration can be found in
e.g., \cite{Cai2009a}, \cite{CaiJ2009}, \cite{Osher2005},
\cite{Yin2008}.

The Bregman iteration relies on the concept of the \textit{Bregman
distance} \cite{Bregman1967}. The Bregman distance of a convex
function $J(\ubf)$ between points $\ubf$ and $\vbf$ is defined as
\begin{equation}
  \label{Bregmandistance}
  B_J^\pbf(\ubf,\vbf) = J(\ubf) - J(\vbf) - \langle \ubf-\vbf, \pbf \rangle,
\end{equation}
where $\pbf \in \partial J(\vbf)$ is some subgradient in the
subdifferential of $J$ at the point of $\vbf$. Clearly,
$B_J^\pbf(\ubf,\vbf)$ is not a distance in the usual sense, since
$B_J^\pbf(\ubf,\vbf) \neq B_J^\pbf(\vbf, \ubf)$ in general. However,
it does measure the closeness between $\ubf$ and $\vbf$ in the sense
that $B_J^\pbf(\ubf,\vbf)\geq 0$ and $B_J^\pbf(\ubf,\vbf)\geq
B_J^\pbf(\wbf,\vbf)$ for all points $\wbf$ on the line segment
connecting $\ubf$ and $\vbf$.

First, consider the following unconstrained optimization problem
\begin{equation}
  \label{optimizationproblemUN} \underset {\ubf} {\textrm{min}} \ \
  J(\ubf)+ H(\ubf),
\end{equation}
where $J(\ubf)$ is some convex function and $H(\ubf)$ is some convex
and differentiable function with ${\textrm{argmin}_\ubf}H(\ubf)=0$.

Instead of directly solving \eqref{optimizationproblemUN}, the
Bregman iteration iteratively solves
\begin{align} \label{BregmanIterationI}
  \ubf^{k+1} & = \underset {\ubf} {\textrm{argmin}} \ \
  B_J^{\pbf^{k}}(\ubf,\ubf^{k})+ H(\ubf), \notag \\ & = \underset {\ubf} {\textrm{argmin}} \ \
  J(\ubf)-J(\ubf^k)- \langle \ubf-\ubf^k, \pbf^k \rangle + H(\ubf),
\end{align}
for $k=0, 1, \ldots,$ starting from $\ubf^0=\zebf$ and
$\pbf^0=\zebf$. In \eqref{BregmanIterationI}, the updating formula
for $\pbf^k$ is based on the optimality conditions of
\eqref{BregmanIterationI}. Since $\ubf^{k+1}$ minimizes
\eqref{BregmanIterationI}, then $0\in
\partial \{B_J^{\pbf^{k}}(\ubf,\ubf^{k})+ H(\ubf)\}$, where this subdifferential is
evaluated at $\ubf^{k+1}$, i.e.,
\begin{equation*}
  0 \in \partial J(\ubf^{k+1}) - \pbf^{k} + \nabla H(\ubf^{k+1}).
\end{equation*}
This leads to
\begin{equation}
  \label{BregmanIterationII} \pbf^{k+1} = \pbf^{k} - \nabla H(\ubf^{k+1}) \in \partial J(\ubf^{k+1}).
\end{equation}
Combining \eqref{BregmanIterationI} and \eqref{BregmanIterationII}
yields the Bregman iteration:
\begin{equation}\label{BregmanforUnconstrained}
  \left\{\begin{array}{l} \ubf^{k+1} = {\textrm{argmin}_\ubf}
    B_J^{\pbf^{k}}(\ubf,\ubf^{k})+ H(\ubf),  \\ \pbf^{k+1} = \pbf^{k} - \nabla H(\ubf^{k+1}),\end{array}  \right.
\end{equation}
for $k=0, 1, \ldots,$ starting with $\ubf^0 = \zebf$ and $\pbf^0 =
\zebf$.

The convergence of the Bregman iteration
\eqref{BregmanforUnconstrained} was analyzed in \cite{Osher2005}. In
particular, it was shown that, under fairly weak assumptions on
$J(\ubf)$ and $H(\ubf)$, $H(\ubf^k)\rightarrow 0$ as $k \rightarrow
\infty$.

We then show that the Bregman iteration can also be used to solve
the general constrained convex minimization problem:
\begin{equation}\label{optimizationproblemC}
 \underset {\ubf} {\textrm{min}} \ \
  J(\ubf) \ \ s.t. \ \ \Phibf\ubf = \ybf,
\end{equation}
where $J(\ubf)$ denotes some convex function and $\Phibf$ is some
linear operator.

Traditionally, this problem may be solved by a continuation method,
where we solve sequentially the unconstrained problems
\begin{equation}\label{GeneraloptmodelUCC}
\underset {\ubf} {\textrm{min}} \ \
  J(\ubf) + \frac{\lambda_k}{2}\|\Phibf\ubf - \ybf\|_2^2,
\end{equation}
where $\lambda_1<\lambda_2<\ldots<\lambda_K$ is an increasing
sequence of penalty function weights \cite{Boyd2004}. In order to
enforce that $\Phibf\ubf \approx \ybf$, we must choose $\lambda_K$
to be extremely large. However, choosing a large value for
$\lambda_k$ may make \eqref{GeneraloptmodelUCC} extremely difficult
to solve numerically \cite{Goldstein2009}.

The Bregman iteration provides another way to transfer the
constrained problem \eqref{optimizationproblemC} into a series of
unconstrained problems. To this end, we first convert
\eqref{optimizationproblemC} into an unconstrained optimization
problem using a quadratic penalty function:
\begin{equation}\label{GeneraloptmodelUC}
\underset {\ubf} {\textrm{min}} \ \
  J(\ubf) + \frac{\lambda}{2}\|\Phibf\ubf - \ybf\|_2^2.
\end{equation}
Then we apply the Bregman iteration \eqref{BregmanforUnconstrained}
and iteratively minimize:
\begin{equation}\label{BregmanforConstrainedI}
  \left\{\begin{array}{l} \ubf^{k+1} =  {\textrm{argmin}_\ubf}
  B_J^{\pbf^{k}}(\ubf,\ubf^{k})+ \frac{\lambda}{2}\|\Phibf\ubf - \ybf\|_2^2,  \\ \pbf^{k+1} = \pbf^{k} - \lambda \Phibf^*(\Phibf\ubf^{k+1} -
\ybf),\end{array}  \right.
\end{equation}
for $k=0, 1, \ldots,$ starting with $\ubf^0 = \zebf$ and $\pbf^0 =
\zebf$.

By change of variable, this seemingly complicated iteration
\eqref{BregmanforConstrainedI} can be reformulated into a simplified
form \cite{CaiJ2009}:
\begin{equation}\label{BregmanforConstrainedIII}
  \left\{\begin{array}{l}  \ubf^{k+1} =  {\textrm{argmin}_{\ubf}}
  J(\ubf) + \frac{\lambda}{2}\|\Phibf\ubf - \ybf + \bbf^k\|_2^2, \\ \bbf^{k+1} = \bbf^k + (\Phibf\ubf^{k+1} - \ybf),\end{array}  \right.
\end{equation}
for $k=0, 1, \ldots,$ starting with $\bbf^0 = 0$ and $\ubf^0 = 0$.

Indeed, by $\pbf^0 = \zebf$ and induction on $\pbf^k$, we obtain
$\pbf^{k} = -\lambda \Phibf^*\sum_{j=1}^{k}(\Phibf\ubf^{j} - \ybf)$.
Substituting this into the first step of
\eqref{BregmanforConstrainedI} yields
\begin{align}
 B_J^{\pbf^{k}}(\ubf,\ubf^{k})+ \frac{\lambda}{2}\|\Phibf\ubf - \ybf\|_2^2 & = J(\ubf) - J(\ubf^k) -
 \langle \ubf-\ubf^k, \pbf^k \rangle + \frac{\lambda}{2}\|\Phibf\ubf - \ybf\|_2^2 \notag \\
 & = J(\ubf) - \langle \ubf, \pbf^k \rangle + \frac{\lambda}{2}\|\Phibf\ubf - \ybf\|_2^2 + C_2 \notag \\
 & = J(\ubf) + \lambda \langle \Phibf\ubf, \sum_{j=1}^{k}(\Phibf\ubf^{j} - \ybf) \rangle + \frac{\lambda}{2}\|\Phibf\ubf - \ybf\|_2^2 + C_2 \notag \\
 & = J(\ubf) + \frac{\lambda}{2}\left\|\Phibf\ubf - \ybf + \sum_{j=1}^{k}(\Phibf\ubf^{j} - \ybf)\right\|_2^2 +
 C_3,
\end{align}
where $C_2$ and $C_3$ are independent of $\ubf$. By the definition
of $\ubf^{k+1}$ in \eqref{BregmanforConstrainedI}, we have that
\begin{equation}\label{equivalentform}
  \ubf^{k+1} = \underset {\ubf} {\textrm{argmin}} \ \
  J(\ubf) + \frac{\lambda}{2}\left\|\Phibf\ubf - \ybf + \sum_{j=1}^{k}(\Phibf\ubf^{j} - \ybf)\right\|_2^2.
\end{equation}
Define $\bbf^k = \sum_{j=1}^k(\Phibf\ubf^k-\ybf)$, then we have
\begin{equation}\label{updateofb}
\bbf^{k+1} = \bbf^k + (\Phibf\ubf^{k+1}-\ybf), \ \ \bbf^0 = \zebf.
\end{equation}
With this, \eqref{equivalentform} becomes
\begin{equation} \label{updateofu}
  \ubf^{k+1} = \underset {\ubf} {\textrm{argmin}} \ \
  J(\ubf) + \frac{\lambda}{2}\|\Phibf\ubf - \ybf + \bbf^k\|_2^2.
\end{equation}
Combining \eqref{updateofu} and \eqref{updateofb} yields
\eqref{BregmanforConstrainedIII}. It is this form
\eqref{BregmanforConstrainedIII} that will be used to derive the
split Bregman iteration.

The convergence results of the Bregman iteration
\eqref{BregmanforConstrainedI} (or \eqref{BregmanforConstrainedIII})
were given in \cite{Goldstein2009}, \cite{Yin2008}. It was shown
that the sequence $\ubf^k$ generated by
\eqref{BregmanforConstrainedI} (or \eqref{BregmanforConstrainedIII})
weakly converges to a solution of \eqref{optimizationproblemC}.

\section*{Acknowledgment}
The authors thank the anonymous referees and the Associate Editor
for useful and insightful comments which have helped to improve the
presentation of this paper. Y. Liu would like to thank Deanna
Needell (with Claremont Mckenna College) and Haizhang Zhang (with
Sun Yat-sen University) for valuable discussions on the topic of
compressed sensing with general frames.

\spacingset{1}
\bibliographystyle{ieeetr}
\bibliography{Sparse Recovery with Coherent and Redundant Dictionaries.bbl}

\end{document}